\documentclass[journal,twoside,web]{ieeecolor}

\usepackage{generic}
\usepackage{cite}
\usepackage{amsmath,amssymb,amsfonts}
\usepackage{hyperref}
\hypersetup{hidelinks}
\usepackage{textcomp}
\usepackage{lipsum}
\usepackage{multirow}
\usepackage{stfloats}
\usepackage{float}
\usepackage{tcolorbox}
\usepackage{bm}
\usepackage{booktabs}

\newtheorem{theorem}{\textbf{Theorem}}
\newtheorem{problem}{\textbf{Problem}}

\newtheorem{definition}{Definition}
\newtheorem{corollary}{Corollary}
\newtheorem{example}{Example}
\newtheorem{assumption}{Assumption}

\newtheorem{lemma}{Lemma}
\newtheorem{remark}{Remark}

\newcommand{\R}{\mathbb{R}}
\newcommand{\N}{\mathbb{N}}
\newcommand{\He}{\operatorname{He}}

\newcommand{\norm}[1]{\left\lVert #1 \right\rVert}

\usepackage{colortbl}
\usepackage{graphics} % for pdf, bitmapped graphics files
\usepackage{epsfig} % for postscript graphics files
\usepackage{color}
\usepackage{doi}
\usepackage{algorithm}
\usepackage{algorithmic}
\usepackage{color,xcolor}

\def\BibTeX{{\rm B\kern-.05em{\sc i\kern-.025em b}\kern-.08em
        T\kern-.1667em\lower.7ex\hbox{E}\kern-.125emX}}
\begin{document}
    
    \title{Active Informativity: Online Input Design \\
        for Data-Driven Control} 
    \author{Yishu Wang, \IEEEmembership{Member, IEEE}, Paolo Rapisarda, \IEEEmembership{Member, IEEE}, \\           Jianquan Lu, \IEEEmembership{Senior Member, IEEE},  Yang Liu, \IEEEmembership{Senior Member, IEEE}
        \thanks{This work was partially supported by the National Natural Science Foundation of China under Grant 62373105, the Natural Science Foundation of Jiangsu Province under Grant BK20240009, and the Jiangsu Provincial Scientific Research Center of Applied Mathematics under Grant BK20233002. \textit{(Corresponding author: Jianquan Lu.)}}
        \thanks{Yishu Wang is with the School of Mathematics, Southeast University, Nanjing 210096, China {(email: wongyisu@gmail.com)}. }
        \thanks{Paolo Rapisarda is with the School of Electronics and Computer Science, University of Southampton, SO17 1BJ, UK (e-mail: pr3@ecs.soton.ac.uk).}
        \thanks{Jianquan Lu is with the School of Mathematics, Southeast University, Nanjing 210096, China {(email: jqluma@seu.edu.cn)}. }
        \thanks{Yang Liu is with the School of Mathematical Sciences, Zhejiang Normal University, Jinhua, 321004, China, and the Hangzhou School of Automation, Zhejiang Normal University, Hangzhou, 311231, China {(email: liuyang@zjnu.edu.cn)}.}
    }
    
    \maketitle
    
    \begin{abstract}
        We develop a unified framework for online input design in data-driven control from noisy data. We introduce a performance measure and analyze how   such quality measure changes as new data are added. Based on such analysis we propose an online method for selecting inputs that ensures that data quality does not decrease over time. We demonstrate the validity of our approach by formulating six common noisy data-driven control problems in our framework, and by solving numerical examples for three of them.
    \end{abstract}
    

    \section{Introduction}\label{sec:introduction}
    
    % Data-driven control approaches for linear systems are fundamentally based on  implicit representations of the plant dynamics based on measurements that are used to devise control inputs  or controller representations. Historically, the starting point for developing such control methods for exact (noiseless) measurements was the fundamental lemma and persistently exciting inputs (see \cite{Willems2005PE,MarkovskyRapisarda2008DataDriven}), that require the data to be {fully informative} about the system dynamics (i.e. that they be  sufficient to guarantee exact identification of the unique data-generating system). Later, data-driven approaches were developed also for the case when several dynamical systems are compatible with the given measurements, whether for lack of full informativity or because the measurements are affected by noise (see the approach of \cite{VanWaarde2023InformativityApproach} based on informativity; \cite{9308978, VanWaarde2023QMI} based on the matrix $S$-lemma; and \cite{BisoffiDePersisTesi2022Petersen} for a robustness-based approach to data-driven control for noisy measurements). A thorough exposition of the state-of-the-art on data-driven control is in the book \cite{DBLSCT2025}. 
    
    \emph{Input design}, i.e. choosing the input signals applied to a dynamic system to  maximize the ``information content" of the resulting output trajectory,  has a long history in system identification (see for example \cite{GoodwinPayne1977,AnnergrenEtAl2017Application}). Recently such techniques  have raised interest also in data-driven control, motivated by  computational- and storage efficiency considerations (see the papers \cite{AlsaltiLopezMuller2023PEInputs,DePersisTesi2021DesigningExperiments,VanWaarde2022BeyondPE,CAMLIBEL2025106045,iannelli21,IANNELLI2021625,BerberichQuantitativeFL} and Chapter 12 of \cite{DBLSCT2025}). \emph{Online input design}   techniques are especially relevant in a data-driven context, since they build on the partial information about the system dynamics contained in the past data to inform the design of subsequent input values. Informally, such techniques aim at solving the following problem: \emph{given a set of available data} consisting of consecutive measurements of the relevant variables (state $x$ or outputs $y$ and inputs  $u$) produced by a linear, time-invariant system \emph{over $[0,T]$, how can one  choose the  value of $u(T+1)$ so that the corresponding data over $[0,T+1]$ is ``better'' for the control objectives?} 
    
    Naturally, the setting and the application determine what makes a set of data ``better'' than another one, thus determining the choice of the next input value. For example, in the case of \emph{noiseless} measurements  one chooses $u(T+1)$ so as to  increase the rank of the Hankel matrix of the data and achieve more quickly the desired level of persistency of excitation (see \cite{VanWaarde2022BeyondPE}). To the best of the authors’ knowledge, experiment design for the case of \emph{noisy} data has received considerably less attention. One exception is \cite{CoulsonQuantitativePE}, where a quantitative notion of ``persistency of excitation" was introduced, related to the smallest singular value of the data matrix. It was experimentally verified therein that a larger  persistency of excitation degree corresponds to a better performance of  data-driven control algorithms. In \cite{CoulsonQuantitativePE} the selection of the input was made \emph{off-line}. In this paper, we study the \emph{online} input design problem in the informativity framework (see \cite{VanWaarde2023InformativityApproach}), assuming a noise model characterized by a quadratic matrix inequality (QMI), as proposed in \cite{VanWaarde2023QMI}. Such noise model captures several situations of practical interest (see Section 2 of \cite{VanWaarde2023QMI}) and under suitable assumptions can be effectively used to solve data-driven control problems by reducing them to the solution of semi-definite programs (SDPs) whose structure (number of variables, constraints, etc.) is problem-dependent (see \cite{VanWaarde2022Dissipativity,VanWaarde2024BehavioralNoisyIO,VanWaardeCamlibelMesbahi2022NoisyData}). 
    
    Our contributions are:
    \begin{enumerate}
        \item In Section \ref{sec:setupfram} we provide a  general framework to formulate the structure of  SDPs arising in data-driven control in the noisy case. 
        
        \item We show in Section \ref{sec:setupfram} and in Appendix \ref{app:3moreex} how such general formulation covers six  data-driven control problems, namely  quadratic stabilization, $\mathcal{H}_2$ control, $\mathcal{H}_\infty$ control, see \cite{VanWaardeCamlibelMesbahi2022NoisyData}; additive- and multiplicative feedback fragility, see \cite{li2025fragilityanalysisdatadrivenfeedback}; and strict dissipativity, see \cite{VanWaarde2022Dissipativity};
        
        \item Within such unifying framework, we define an \emph{informativity measure} associated to the data available at time $T$ (Section \ref{sec:setupfram}), to capture the quality of the data with respect to the control objective;  
        \item In Section \ref{sec:2keylemmas} we show how the value of such informativity measure for the data over $[0,T+1]$ is related to its value for the data over $[0,T]$;  
        \item We exploit such relation to develop a general online input design procedure (see Sections \ref{sec:SOCP},  \ref{sec:SDP}) to select at time $T+1$ an input value that improves the quality of the data at time $T+1$. 
    \end{enumerate}
    
    The results presented here support also \emph{methodologically} the informativity framework for the noisy case. In such framework controllers are computed solving SDPs involving the data. Such SDPs have a fixed structure and often the user has limited experience of how to choose the parameters of the LMI solver to achieve a better result. The selection of data is consequently an expedient way to improve the performance of the controller. In this paper, we show that assuming some prior knowledge on the noise, the choice of inputs can be done online and that it considerably improves the outcome of the data-driven controller design procedure (on this issue, see the last sentence of \cite{CoulsonQuantitativePE}). Admittedly, given that such performance improvement comes at the cost of additional computation and storage requirement,  adopting  the methodology proposed here involves practical trade-offs that must be evaluated case-by-case. In the numerical examples presented in this paper the online procedure considerably outperforms a random choice of inputs.  
    
    The paper is organized as follows: in Section \ref{sec:setup} we define the notation, recall the notion of data-consistent system, and provide two descriptions (one QMI-, the other matrix-ellipsoid-based) of such systems. In Section  \ref{sec:genfram} we formulate a general form of the LMIs involved in data-driven control problems, and we prove two instrumental results. In Section \ref{sec:perfimpro} we illustrate the two optimization sub-problems on which our input procedure is based. The procedure is formally stated in Section \ref{sec:algo}, where we apply it to three exemplary problems. Appendix \ref{app:3moreex} and Appendix \ref{app:3morappl} contain further examples and applications, to emphasize the generality of the framework we propose. Appendix \ref{app:cover} contains some background information related to covering of bounded compact sets to make the paper self-contained. 
    
    \subsection*{Notation.}
    
    $\N$ and $\R$ denote respectively the set of natural and real  numbers; $\R^n$ denotes the space of $n$-dimensional vectors with real entries; and  $\R^{n\times m}$ denotes the set of $n\times m$ matrices with real entries. The transpose of a matrix $M$ is denoted by $M^\top$.  
    Given a matrix $M$, $\operatorname{He}(M):=M+M\top$, and $M^\dagger$ denotes the Moore–Penrose inverse. 
    Given $Q=Q^\top\in\R^{n\times n}$, $Q\succ 0$ denotes positive-definiteness and $Q\succeq 0$ semi-definite positiveness. If $Q\succeq 0$, then $\lambda_{\mathrm min}(Q)$ its smallest eigenvalue, $Q^\frac{1}{2}$ denotes any matrix such that $Q^\frac{1}{2} Q^\frac{1}{2} =Q$. If $Q=Q^\top\in\R^{n\times n}$ and $v\in\R^n$, $\|v \|_Q^2:=v^\top Q v$. The Hermitian part of a matrix expression is denoted by $\He$.

    \section{Background material}\label{sec:setup}
    We recall the definitions and notation from the literature that are most relevant for the illustration of the  results presented in this paper. We refer the reader to \cite{VanWaarde2023QMI,CAMLIBEL2025106045,VanWaardeCamlibelMesbahi2022NoisyData} and \cite{BisoffiDePersisTesi2022Petersen} for a thorough exposition. 
    
    \subsection{Data-consistent systems}
    The unknown discrete-time data-generating linear system has the representation 
    \begin{align}\label{eq:truesyst}
        x_{t+1}=A_\star x_t+B_\star u_t+e_t,
    \end{align}
    where $x_t\in\mathbb R^n$, $u_t\in\mathbb R^m$, and $e_t\in\mathbb R^n$ is an unmeasured process disturbance. In the rest of this paper we assume that the system $(A_\star,B_\star)$ is controllable but known, and that only the input and the state variable are available for measurements; from such data and the unmeasured noise we define the matrices 
    \begin{align}\label{eq:symboldef}
        X_T&=[x_0\ \cdots\ x_{T-1}]\in\mathbb R^{n\times T},\nonumber\\ 
        U_T&=[u_0\ \cdots\ u_{T-1}]\in\mathbb R^{m\times T},\nonumber\\ 
        Y_T&=[x_1\ \cdots\ x_T]\in\mathbb R^{n\times T},\nonumber\\
        E_T&=[e_0\ \cdots\ e_{T-1}]\in\mathbb R^{n\times T}\;.
        %    W_T&=\begin{bmatrix}X_T\\U_T\end{bmatrix}\in\mathbb R^{p\times T}\; .
    \end{align}
    In the rest of the paper we denote 
    \begin{equation}\label{eq:pdef}
            W_T:=\begin{bmatrix}X_T\\U_T\end{bmatrix}\in\mathbb R^{p\times T}, \qquad p:=n+m\; .
    \end{equation} 
    We make the following assumption.\medskip 
    
    \begin{assumption}[Full row-rank data-matrix]\label{ass:fullrank}
        The measurements matrix $W_T$ has full row-rank:
        \[
        \operatorname{rank}(W_T)=p\; .
        \]
    \end{assumption}\medskip 
    We denote $
    \Theta_\star:=[A_\star\ B_\star]\in\mathbb R^{n\times p}$; it follows from the positions \eqref{eq:symboldef}--\eqref{eq:pdef} that 
    \begin{align*}
        Y_T=\Theta_\star W_T+E_T\; .
    \end{align*}
    We make the following assumption on the noise sequence. \medskip 
    \begin{assumption}[Energy-bound on noise]\label{ass:noisebound} There exists a known matrix sequence $\{\Delta_t\in\R^{n\times n} \}$ 
        such that for every $T\in\N$ the matrix $E_T$ defined in \eqref{eq:symboldef} corresponding to the noise sequence $\{e_t\}_{t=0,\ldots,T-1}$ belongs to the set $\mathcal E_T$ defined by
        \begin{align}\label{eq:noiseset}
            \mathcal E_T:=\left\{E\in\mathbb R^{n\times T}: E E^\top\preceq\Delta_T\right\}.
        \end{align}
    \end{assumption}\medskip 
    
    \begin{remark}\label{rem:QMI}\rm
        The case of a uniform bound $e_t e_t^\top \preceq \epsilon I$ for all $t\in\N$ falls under Assumption \ref{ass:noisebound}; choose $\Delta_T=T \epsilon I_n$. 
        
        In \cite{DBLSCT2025} QMI descriptions 
        \[
        \begin{bmatrix}
            I\\ E_T^\top
        \end{bmatrix}^\top \begin{bmatrix}
            \Phi_{11}&\Phi_{12}\\ \Phi_{12}^\top& \Phi_{22}
        \end{bmatrix} \begin{bmatrix}
            I\\ E_T^\top
        \end{bmatrix}\succeq 0\; ,
        \]
        of the noise properties have been considered, see Section 3.4 therein. The assumption on noise-boundedness represented by the matrix inequality in \eqref{eq:noiseset} corresponds to the case $\Phi_{12}=0$, $\Phi_{11}=\Delta_T$, and $\Phi_{22}=-I$. Note that in this case the matrix $\Phi$ depends on the number of available data samples. \hfill $\blacksquare $ 
    \end{remark}\smallskip
    
    In some data-driven control problems it is necessary to formalize \emph{input constraints}; we do this defining a set $\mathcal U$ that may include  several  constraints such as norm-boundedness, linear constraints, or distance from a given signal:
    \begin{align}\label{eq:mathcalU}
        \begin{cases}
            \|u\|_\infty\le\bar u_\infty,\\
            \|u\|_2\le \bar u_2,\\
            (u-u_c)^{\top} H_u^{-1}(u-u_c)\le1,\\
            Gu\le g.
        \end{cases}
    \end{align}
    The following assumption is satisfied for the examples above. 
    \begin{assumption}[Compactness of $\mathcal{U}$]\label{ass:compact}
        The set $\mathcal{U}$ is compact. 
    \end{assumption}
    %\textcolor{red}{YISHU: we need this assumption for the cover, see after Lemma \ref{lemma3}}
    \medskip
    
    Given (state,input) data $\{x_t\}$ and $\{u_t\}$ generated by the ``true" system $(A_\star,B_\star)$ as in \eqref{eq:truesyst} together with the noise model \eqref{eq:noiseset}, there exist many pairs $(A,B)$ such that 
    \[
    Y_T=\begin{bmatrix}
        A&B
    \end{bmatrix} W_T+E_T \; , \; E_T\in\mathcal E_T\; .
    \]
    We call each such pair, corresponding to the matrix 
    \[
    \Theta:=\begin{bmatrix}
        A&B
    \end{bmatrix}\; ,
    \]
    a \emph{data-consistent system}. 
    We denote the set of {data-consistent systems} by 
    \begin{align}\label{eq:dataconsist}
        \Sigma_T := \left\{ \Theta\in\mathbb R^{n\times p}:(Y_T-\Theta W_T)(Y_T-\Theta W_T)^{\top}\preceq\Delta_T\right\}.
    \end{align}
    In the rest of this paper we use two descriptions of $\Sigma_T$; the first one is the QMI-based one pioneered by \cite{VanWaardeCamlibelMesbahi2022NoisyData}, and the second one is the matrix-ellipsoid introduced in \cite{BisoffiDePersisTesi2021Tradeoffs}. We briefly recall them in the next two subsections. 
    
    \subsection{QMI description of $\Sigma_T$}
    Following the discussion in Remark \ref{rem:QMI}, we define 
    \begin{align*}
        \Phi_T := \begin{bmatrix}
            \Delta_T & 0 \\
            0 & -I_T
        \end{bmatrix}, \quad
        G_T := \begin{bmatrix}
            I_n & Y_T \\
            0 & -X_T \\
            0 & -U_T
        \end{bmatrix}\; ,
    \end{align*}
    and 
    \begin{align}
        N_T&:=G_T\Phi_TG_T^\top\notag\\ 
        &= \begin{bmatrix}
            I_n \\
            0 \\
            0
        \end{bmatrix} \Delta_T \begin{bmatrix}
            I_n \\
            0 \\
            0
        \end{bmatrix}^\top - \begin{bmatrix}
            Y_T \\
            -X_T \\
            -U_T
        \end{bmatrix} \begin{bmatrix}
            Y_T \\
            -X_T \\
            -U_T
        \end{bmatrix}^\top. \label{eq13}
    \end{align}
    It is a matter of straightforward verification to check that the following equivalence holds: 
    \begin{align}
        \Theta\in\Sigma_T \quad\Longleftrightarrow\quad \begin{bmatrix}I_n\\\Theta^{\top}\end{bmatrix}^{\top} N_T \begin{bmatrix}I_n\\\Theta^{\top}\end{bmatrix} \succeq0.
    \end{align}
    % This is exactly the data QMI shared by all six kinds of performance in the Matrix S-Lemma framework.

    \subsection{Matrix ellipsoid description of $\Sigma_T$}
    
    Define 
    \begin{align}\label{eq:ssellipsoid}
        \mathbf A_T&:=W_TW_T^{\top}\; ,\nonumber \\
        \mathbf B_T&:=-W_TY_T^{\top}\; ,\\
        \mathbf C_T&:=Y_TY_T^{\top}-\Delta_T\; .\nonumber
    \end{align}
    Under Assumption \ref{ass:fullrank} $\mathbf A_T$ is invertible (and consequently positive-definite); define  
    \begin{align}\label{eq:zeta&QT}
        \bm \zeta_T&:=-\mathbf A_T^{-1}\mathbf B_T\in\mathbb R^{p\times n},\nonumber\\
        \mathbf Q_T&:=\mathbf B_T^{\top}\mathbf A_T^{-1}\mathbf B_T-\mathbf C_T\succeq0.
    \end{align}
    It can be verified that the data-consistent system set $\Sigma_T$ can be written as the following matrix ellipsoid (see \cite{BisoffiDePersisTesi2022Petersen}):
    \begin{align}\label{eq:matrellips}
        \Sigma_T = \left\{ \Theta=Z^{\top}: Z=\bm \zeta_T+\mathbf A_T^{-1/2}\Upsilon \mathbf Q_T^{1/2},\ \|\Upsilon\|\le1 \right\}.
    \end{align}

    % \textcolor{red}{{\bf YISHU}- I did not raise this before but I need to now: why do you use a ``$Z$", a ``$Z^\top$", etc.?}
    
    % \textcolor{blue}{{\bf PAOLO}- This presentation is chosen to be consistent with (16) and (19) in \cite{BisoffiDePersisTesi2022Petersen}. The result was first obtained therein, and I therefore follow that formulation.}
    
    \section{A unifying Framework}\label{sec:genfram}
    We introduce in Section \ref{sec:setupfram} a parametric version of the LMIs involved in many data-driven control problems, and the definition of a performance index related to its solutions. In Section \ref{sec:2keylemmas}  we state two instrumental results to prove the main results of the paper. 
    \subsection{Set-up}\label{sec:setupfram}

    % All six criteria have the same structure when the data are updated:
    % \begin{align}
        %     \mathcal C_{T+1}=\mathcal C_T+\alpha hh^{\top}-\alpha R.
        % \end{align}
    % where
    % \begin{align}
        %     \mathcal C_{T} = &\text{ data-independent certificate part}  \notag\\
        %     &- \alpha\times\text{common data QMI}.
        % \end{align}
    
    % \subsection{Matrix S-Lemma criterion}\label{sec:genframSLemma}
    
    We denote by $\vartheta:= (\vartheta_1,\vartheta_2,\ldots,\vartheta_s)$  the decision variables involved in the LMIs;  by $\alpha\ge0$ the matrix S-Lemma multiplier; and by $\rho$  an index measuring ``performance", with larger values being associated with better performance. Many informativity conditions appearing in data-driven control problems in the informativity setting can be written in a unified form (the subscript $T$ indicates the number of data samples) as 
    \begin{align}\label{eq18}
        {\mathcal C_T(\vartheta,\alpha,\rho) := \mathcal F(\vartheta,\rho)-\alpha ON_TO^{\top}\succeq 0\; .} 
    \end{align}
    Here, $O$ is a zero-padding constant matrix; $\mathcal F$ is data-independent and is affine in the decision variables $\vartheta$ and $\rho$, and $N_T$ is defined in \eqref{eq13}. 
    
    We express all \emph{data-independent constraints} (positive-definiteness, trace, performance, and scaling-normalization constraints) by defining a set $\mathcal{K}$ and imposing that 
    \begin{align}
        (\vartheta,\alpha,\rho)\in\mathcal K\; .
    \end{align}
    \begin{definition}
        The \emph{performance index} associated with the first $T$ data samples is 
        \begin{align}
            \rho_T^\star := \sup\left\{ \rho: \exists\ (\vartheta,\alpha,\rho)\in\mathcal K \text{ such that } \mathcal C_T(\vartheta,\alpha,\rho)\succeq0 \right\}. \label{eq20}
        \end{align}
    \end{definition}
    We exemplify such concepts using three standard data-driven control problems, see \cite{DBLSCT2025} for more details. We illustrate three other examples in  Appendix \ref{app:3moreex}. 
    
    \begin{example}[Quadratic stabilization]\label{ex:QS}\rm We choose as decision variables the pair $\vartheta=(P,L)$ and define
        \begin{align*}
            \mathcal F(P,L,\rho) &:= \begin{bmatrix}
                P & 0 & 0 & 0 \\
                0 & -P & -L^{\top} & 0 \\
                0 & -L & 0 & L \\
                0 & 0 & L^{\top} & P
            \end{bmatrix} -\rho I_{3n+m} \\
            O &:= \begin{bmatrix}
                I_{2n+m} \\
                0_{n\times(2n+m)}
            \end{bmatrix}\; .
        \end{align*}
        To remove the homogeneous scaling in the criterion, we define 
        \begin{align*}
            \mathcal K_{stab} := \left\{ (P,L,\alpha,\rho): P\succeq0,\ \alpha>0,\ \operatorname{tr} P+\alpha=1 \right\}\; .
        \end{align*}
        The corresponding performance is the normalized strict LMI margin defined by 
        \begin{align*}
            \rho_T^\star:= \sup\left\{ \rho: \exists\  (P,L,\alpha,\rho)\in\mathcal K_{\rm stab}\text{ such that } \right.\\
            \left.\mathcal F(P,L,\rho)-\alpha ON_TO^{\top}\succ 0\right\}.
        \end{align*}
        
        The value $\rho_T^\star$ is the quadratic-stabilization certificate margin of $\Sigma_{T}$. Its positivity is
        equivalent to the existence of a strict common quadratic stabilization certificate for all systems in $\Sigma_{\mathcal D_T}$. If such a condition holds, then for every $(P, L)$ for which $\exists\ \alpha$ s.t. $(P,L,\alpha,\rho^\star)\in\mathcal K_{stab}$ it holds that $K=LP^{-1}$ is a stabilizing feedback gain for all systems in $\Sigma_{T}$. \hfill $\blacksquare $
    \end{example}
    \smallskip
    
    \begin{example}[$\mathcal H_2$ control]\label{ex:H2}\rm 
        Let the performance output be $z=Cx+Du\in\mathbb R^\ell$ and define $C_{Y,L}:=CY+DL$. Set $\rho=-\gamma_2^2$, take $\vartheta=(Y,L,Z)$, and define
        \begingroup
        \small
        \setlength{\arraycolsep}{4pt}
        \begin{align*}
            \mathcal F(Y,L,\rho) &:= \begin{bmatrix}
                Y & 0 & 0 & 0 & 0 \\
                0 & 0 & 0 & Y & 0 \\
                0 & 0 & 0 & L & 0 \\
                0 & Y & L^{\top} & Y & C_{Y,L}^{\top} \\
                0 & 0 & 0 & C_{Y,L} & I_\ell
            \end{bmatrix}, \\
            O &= \begin{bmatrix}
                I_{2n+m} \\
                0_{(n+\ell)\times(2n+m)}
            \end{bmatrix}.
        \end{align*}
        \endgroup
        Let
        \begin{align}
            &\mathcal K_{\mathcal H_2} := \Bigg\{ (Y,L,Z,\alpha,\rho):
            \ Y=Y^{\top}\succeq0,\ Z=Z^{\top},\ \alpha>0,\notag\\
            &\rho\leq0,\ 
            \begin{bmatrix}
                Y & C_{Y,L}^{\top} \\
                C_{Y,L} & I_\ell
            \end{bmatrix}\succeq0,\ 
            \begin{bmatrix}
                Z & I_n \\
                I_n & Y
            \end{bmatrix}\succeq0,\ 
            \operatorname{tr} Z+\rho\leq0
            \Bigg\}.
        \end{align}
        The corresponding optimal certifiable task level is
        \begin{align}
            \rho_T^\star = \sup\left\{\rho:\exists\ (Y,L,Z,\alpha,\rho)\in\mathcal K_{\mathcal H_2},\right. \notag\\
            \left. \mathcal F(Y,L,\rho)-\alpha ON_TO^{\top}\succ 0\right\}.
        \end{align}
        The value $\rho_T^\star$ is the supremal certifiable $\mathcal H_2$ task level. Equivalently, $-\rho_T^\star$ is the infimum of the squared common $\mathcal H_2$ performance bounds certified from the data. Thus, increasing $\rho_T^\star$ decreases the best certified $\mathcal H_2$ bound.\hfill $\blacksquare $
    \end{example}
    \smallskip

    \begin{example}[Dissipativity]\label{ex:Diss}\rm 
        Let $z=Cx+Du\in\mathbb R^\ell$ and fix the quadratic supply rate
        \begin{align}
            s(u,z) = \begin{bmatrix}u\\z\end{bmatrix}^{\top} S \begin{bmatrix}u\\z\end{bmatrix}, \qquad S=S^{\top}.
        \end{align}
        Assume that $S$ is nonsingular and satisfies the standard inertia conditions of the dissipativity criterion. Partition
        \begin{align}
            -S^{-1} = \begin{bmatrix}
                \widehat F & \widehat G \\
                \widehat G^{\top} & \widehat H
            \end{bmatrix},
        \end{align}
        where $\widehat F\in\mathbb S^m$ and $\widehat H\in\mathbb S^\ell$. Take $\vartheta=Q$ and define
        \begin{align}
            \mathcal F(Q,\rho) &:= \begin{bmatrix}
                Q & 0 & 0 & 0 \\
                0 & \widehat H & 0 & -\widehat G^{\top} \\
                0 & 0 & -Q & 0 \\
                0 & -\widehat G & 0 & \widehat F
            \end{bmatrix} -\rho I_{2n+m+\ell}, \\
            O &= \begin{bmatrix}
                I_n & 0 & 0 \\
                0 & -C & -D \\
                0 & I_n & 0 \\
                0 & 0 & I_m
            \end{bmatrix}.
        \end{align}
        Let
        \begin{align}
            \mathcal K_{\rm diss} := \left\{ (Q,\alpha,\rho): Q=Q^{\top}\succ0,\ \alpha>0,\ \rho\in\mathbb{R} \right\}.
        \end{align}
        The corresponding optimal performance level is
        \begin{align}
            \rho_T^\star =\sup\left\{\rho:\exists\ (Q,\alpha,\rho)\in\mathcal K_{\rm diss},\right.\notag\\ 
            \left. \mathcal F(Q,\rho)-\alpha ON_TO^{\top}\succeq0\right\}.
        \end{align}
        The value $\rho_T^\star$ is the uniform strict-dissipativity certificate margin associated with the fixed supply matrix $S$. Its positivity guarantees a common quadratic storage function that establishes strict dissipativity for every system in $\Sigma_T$. \hfill $\blacksquare $
    \end{example}
    \medskip
    
    Let $\mathcal D_T=(W_T, Y_T, \Delta_T)$ denote the current data record. The main problem investigated in our paper is how to improve  performance by designing the input based on the current data. To formalize such problem, we need the following definition.
    \begin{definition}[Online input policy]
            An online input policy at time $T$ is a map $\pi_T$ that assigns to each data record $\mathcal D_T$ an admissible  input
            \begin{align*}
                u_T=\pi_T(\mathcal D_T)\in\mathcal U.
            \end{align*}
    \end{definition}
    \medskip 

    \begin{problem}[Input design]
            Given a control problem with fixed $\mathcal F$, $O$, and $\mathcal K$, and a current data record $\mathcal D_T$, find an online input policy $\pi_T$ such that $u_T=\pi_T(\mathcal D_T)$ yields
            \begin{align*}
                \rho^\star_{T+1}\ge\rho^\star_T,
            \end{align*}
            for any  state $x_{T+1}=\Theta\begin{bmatrix} x^\top_T & u^\top_T \end{bmatrix}^\top+e$,  $\Theta\in\Sigma_T$ and $Y_T-\Theta W_T= \begin{bmatrix}
                E_T & e
            \end{bmatrix}\in\mathcal E_{T+1}$. 
    \end{problem}
    
    \subsection{Two key lemmas}\label{sec:2keylemmas}
    We establish two results instrumental to the development of our input design procedure. In the first one we relate the unifying formulation \eqref{eq18} for the informativity condition with $T+1$ data samples to the one with $T$ samples. To this purpose, we make the following additional assumption.
    \begin{assumption}[Bound on individual noise samples]\label{ass:indivbound}
        For each $T$ there exists a known matrix $\Omega_T=\Delta_{T+1}-\Delta_T$ such that the disturbance at time $T$ satisfies 
        \[
        e_Te_T^\top \preceq\Omega_T\; .
        \]
    \end{assumption}
    \smallskip
    
    Denote the input at time $T$ by $u:=u(T)$; note that at time $T$ the state value $x_T:=x(T)$ is known and that  $x_{T+1}:=x(T+1)$ depends on it and on $u$. Recall that the symbol $\mathcal{U}$ denotes the set of admissible inputs, described by any combination of the constraints \eqref{eq:mathcalU}. For $u\in\mathcal U$, denote 
    \begin{align}
        w_T(u):=\begin{bmatrix}x_T\\u\end{bmatrix}\in\mathbb R^p \quad \mbox{\rm and } \quad y:=x_{T+1}.
    \end{align}
    
    Define
    \begin{align}\label{eq:h&R}
        q_T(u,y)&:=\begin{bmatrix}y\\-w_T(u)\end{bmatrix}\in\mathbb R^{n+p}, \qquad R_T^N:= \begin{bmatrix}
            \Omega_T & 0 \\
            0 & 0_p
        \end{bmatrix},\nonumber\\
        h_T(u,y)&:=Oq_T(u,y),\qquad R_T:=OR_T^NO^{\top}.
    \end{align}
    \begin{lemma}\label{lemma1}
        Assume the informativity condition at time $T$ is of the form \eqref{eq18}. Define $h_T(u,y)$ and $R_T$ by \eqref{eq:h&R}. Then for any  $(\vartheta,\alpha,\rho)\in\mathcal{K}$,
        \begin{align}
            \mathcal C_{T+1}(\vartheta,\alpha,\rho) = \mathcal C_T(\vartheta,\alpha,\rho) +\alpha h_T(u,y)h_T(u,y)^{\top} -\alpha R_T. \label{eq33}
        \end{align}
    \end{lemma}

    \begin{proof}
        Define $N_T$ as in \eqref{eq13}. It is straightforward to verify that for any $u$ and $y$, the following equation holds:
        \begin{align}
            N_{T+1}=N_T-q_T(u,y)q_T(u,y)^{\top}+R_T^N. \label{eq34}
        \end{align}
        Substituting \eqref{eq34}  into \eqref{eq18} and recalling that $\mathcal F$ is data-independent gives the result.
    \end{proof}
    \medskip 
    
    % \begin{corollary}
        %     If $\Omega_T=0$, then $R_T=0$. Any certificate that is currently feasible remains feasible after one more data point is added. Therefore, for any input and any future observation consistent with the data,
        %     \begin{align}
            %         \boxed{\rho_{T+1}^\star\ge\rho_T^\star.}
            %     \end{align}
        % \end{corollary}
    % By \eqref{eq33}, the new term is $\alpha h h^{\top}\succeq0$, while all other terms remain unchanged. In other words, without disturbances, performance does not decrease as more data are collected.
    
    \begin{corollary}
        If $e_T=0$, then  $        \rho_{T+1}^\star\ge\rho_T^\star$. 
    \end{corollary}
    
    \begin{proof}
        If there is no disturbance at time $T$ then   $\Omega_T=0$, and consequently $R_T=0$, see \eqref{eq:h&R}. Any certificate that is  feasible with $T$ data samples remains feasible after one more data point. Moreover, from \eqref{eq33} it follows that for any input and any future observation consistent with the data, $\mathcal{C}_{T+1}(\vartheta,\alpha,\rho)=\mathcal{C}_{T}(\vartheta,\alpha,\rho)+\alpha h h^{\top}\succeq0$, so the value of the performance index at time $T+1$ cannot decrease.
    \end{proof}
    \medskip

    \begin{remark}
            Lemma~\ref{lemma1} shows that collecting more data does not necessarily enlarge the LMI feasible set. To illustrate this, consider Example~\ref{ex:QS}. Repeated application of \eqref{eq34} gives, for any integer $T_0\ge1$,
            \begin{align*}
                N_{T+T_0}=&N_T- \sum_{i=0}^{T_0-1}\bigl(q(u_{T+i},x_{T+1+i})q(u_{T+i},x_{T+1+i})^{\top}\\
                &-R_{T+i}^N\bigr).
            \end{align*}
            With all other constraints unchanged, if the sum term on the right-hand side is positive semidefinite, every previously feasible solution remains feasible, and the corresponding quadratic stabilization guarantee is preserved. If the sum term is negative semidefinite, the inclusion is reversed. Thus, the effect of additional data on the LMI feasible set depends on how the contribution of the new data matrix compares with the increase in the noise bound. The numerical example in Section~\ref{sec:two-state system} confirms this conclusion.
            \hfill $\blacksquare$
    \end{remark}
    \medskip
    
    In the second instrumental result of this section we parametrize the states $x_{T+1}$ as a function of the input $u$ at time $T$ and the disturbance $e_T$.
    \begin{lemma} \label{lemma2}
        Recall the definitions \eqref{eq:ssellipsoid} and \eqref{eq:zeta&QT}. Let $u\in \mathcal{U}$, and define 
        \begin{align*}
            \bar y_T(u):=\bm \zeta_T^{\top} w_T(u) \quad \mbox{\rm and} \quad 
            \sigma_T^2(u):=w_T(u)^{\top}\mathbf A_T^{-1}w_T(u)\; .
        \end{align*}
        Let  $e$ satisfy Assumption \ref{ass:indivbound}; the following statements are equivalent:
        \begin{enumerate}
            \item There exist $\Theta\in\Sigma_T$ such that $y=\Theta w_T(u)+e$; 
            \item There exist $v$ such that $\|v\|\le\sigma_T(u)$ and 
            \begin{align}
                y=\bar y_T(u)+\mathbf Q_T^{1/2}v+e,  \label{eq36}
            \end{align}
        \end{enumerate}
    \end{lemma}
    
    \begin{proof}
        We prove the implication $1) \Longrightarrow 2)$. Since $\Theta$ is a data-consistent system it can be represented in the matrix ellipsoid form \eqref{eq:matrellips}; consequently $\Theta=Z^\top$ where $Z$ is of the form 
        \begin{align*}
            Z=\bm\zeta_T+\mathbf A_T^{-1/2}\Upsilon\mathbf Q_T^{1/2} 
        \end{align*}
        for some  $\Upsilon$ such that $\norm{\Upsilon}_2\leq 1$. Define $r:=\mathbf A_T^{-1/2}w_T(u)$ and $v:=\Upsilon^\top r$. Then
        \begin{align*}
            \Theta w_T(u)&=\bm\zeta_T^\top w_T(u)+\mathbf Q_T^{1/2}\Upsilon^\top\mathbf A_T^{-1/2}w_T(u)\\
            &=\bar y_T(u)+\mathbf Q_T^{1/2}v,
        \end{align*}
        and $\norm{v}_2\leq\norm{\Upsilon}_2\norm{r}_2\leq\sigma_T(u)$. Thus, $y$ is of the form \eqref{eq36}.
        
        We prove the implication $2) \Longrightarrow 1)$. Choose $v\in\R^n$ such that $\norm{v}_2\leq\sigma_T(u)$ and set $r:=\mathbf A_T^{-1/2}w_T(u)$. If $r=0$, then $v=0$, and we define $\Upsilon:=0$. If $r\neq0$, we define
        \begin{align*}
            \Upsilon:=\frac{rv^\top}{\norm{r}_2^2}.
        \end{align*}
        Then $\Upsilon^\top r=v$ and $\norm{\Upsilon}_2=\norm{v}_2/\norm{r}_2\leq1$. In either case, define
        \begin{align*}
            Z:=\bm\zeta_T+\mathbf A_T^{-1/2}\Upsilon\mathbf Q_T^{1/2},\qquad \Theta:=Z^\top.
        \end{align*}
        The matrix ellipsoid representation gives $\Theta\in\Sigma_T$; moreover 
        \begin{align*}
            \Theta w_T(u)=\bar y_T(u)+\mathbf Q_T^{1/2}\Upsilon^\top r=\bar y_T(u)+\mathbf Q_T^{1/2}v.
        \end{align*}
        This concludes the proof of the Lemma. 
    \end{proof}
    \medskip

    \section{Input design for performance improvement} \label{sec:perfimpro}
    Our algorithm for input-design is based on a two-stage computation. In the first stage, illustrated in Section \ref{sec:SOCP}, we solve a second-order cone program (SOCP) to compute from  the current certificate  a vector along which the increase of performance value is the smallest. We call such vector the ``weakest direction". In Section \ref{sec:SDP} we illustrate the second stage, where we use the weakest direction to compute an input that improves the optimal value of the certificate. 
    \subsection{Unified SOCP}\label{sec:SOCP}
    
    If the set 
    \[
    \left\{ (\vartheta,\alpha,\rho)\in\mathcal K \text{ such that }  \mathcal C_T(\vartheta,\alpha,\rho)\succeq 0\right\}\neq \varnothing\; ,
    \]
    then from the definition \eqref{eq20} of optimal performance index $\rho_T^\star$ it follows that for any $\varepsilon_T>0$,  any value $\bar\rho_T$ that satisfies 
    \begin{align}
        \rho_T^\star\ge \bar\rho_T\ge\rho_T^\star-\varepsilon_T \label{eq:barrhoT}\; ,
    \end{align}
    also corresponds to some \emph{feasible certificate}, i.e. a triple 
    $(\bar\vartheta_T,\bar\alpha_T,\bar\rho_T)$ 
    such that
    \begin{align}
        \mathcal C_T(\bar\vartheta_T,\bar\alpha_T,\bar\rho_T)\succeq0, \qquad (\bar\vartheta_T,\bar\alpha_T,\bar\rho_T)\in\mathcal K. \label{eq40}
    \end{align}
    Denote by $\lambda_T^{\min}$ the minimum eigenvalue of the current certificate, and by $\mathcal V_T$ the associated eigenspace: 
    \begin{align*}
        \lambda_T^{\min} &:= \lambda_{\min}(\mathcal C_T(\bar\vartheta_T,\bar\alpha_T,\bar\rho_T)),\\
        \mathcal V_T &:= \ker\bigl(\mathcal C_T(\bar\vartheta_T,\bar\alpha_T,\bar\rho_T)-\lambda_T^{\min}I\bigr).
    \end{align*}
    We call $\mathcal V_T$ is the \emph{weakest eigenspace} of the current certificate matrix. Every vector in such set corresponds to the worst performance attainable with the given certificate matrix. 
    
    We denote by $\bar z_T$ the \emph{weakest direction of the current certificate}, defined as follows. If $\dim\left(\mathcal{V}_T\right)=1$, we define $\bar z_T$ the unit-norm eigenvector associated with $\mathcal{V}_T$. If $\dim\left(\mathcal{V}_T\right)>1$, we use the following deterministic selection rule:
    \begin{align}
        \bar z_T \in \arg\max_{z\in\mathcal V_T,\ \|z\|=1} \|O^{\top} z\|. \label{eq42}
    \end{align}
    \begin{remark}[Computation of $\bar z_T$]\rm 
        Computing the set of solutions to \eqref{eq42} is a finite-dimensional eigenvalue problem on $\mathcal V_T$. Indeed, if the columns of $V_T$ form an orthonormal basis of $\mathcal V_T$, we may choose
        \begin{align*}
            \bar z_T=V_T\xi_T,
        \end{align*}
        where $\xi_T$ is a unit eigenvector associated with the largest eigenvalue of $V_T^{\top} OO^{\top} V_T$. 
        \hfill $\blacksquare $
    \end{remark}
    \medskip
    % \textcolor{red}{{\bf PAOLO:} I add a remark below to explain how $\bar{\rho}_T$ is obtained without computing the supremum in \eqref{eq20} exactly.}

    \begin{remark}[Computation of $\bar \rho_T$]\label{rem:brhoT}\rm 
            We show that to compute a feasible value $\bar\rho_T$ satisfying \eqref{eq:barrhoT} there is no need to compute the supremum in \eqref{eq20} to verify that the inequalities are satisfied.  
            Fix $\varepsilon_T>0$ and choose a certificate satisfying \eqref{eq40} with $\bar\alpha_T>0$. Test the feasibility of the problem 
            \begin{align*}
                (\vartheta,\alpha,\rho)\in\mathcal K,\quad
                \mathcal C_T(\vartheta,\alpha,\rho)\succeq0,\quad
                \rho\ge\bar\rho_T+\varepsilon_T.
            \end{align*}
            If such problem is infeasible, then
            $\bar\rho_T\le\rho_T^\star\le\bar\rho_T+\varepsilon_T$,
            so \eqref{eq:barrhoT} holds.
            Otherwise, choose a feasible solution
            $(\vartheta^+,\alpha^+,\rho^+)$, update
            \begin{align*}
                (\bar\vartheta_T,\bar\alpha_T,\bar\rho_T)\leftarrow
                \tfrac12(\bar\vartheta_T+\vartheta^+,\bar\alpha_T+\alpha^+,\bar\rho_T+\rho^+),
            \end{align*}
            and repeat. By convexity, each update preserves feasibility and the corresponding $\bar\alpha_T>0$, while increasing $\bar\rho_T$ by at least
            $\varepsilon_T/2$. Since $\rho_T^\star$ is finite, the procedure terminates after finitely many updates.
            \hfill $\blacksquare$
    \end{remark}
    %\textcolor{red}{YISHU: I think what is said in this remark is correct.}
    \medskip
    
    Recall the definition of $O$ as the zero-padding matrix associated with the certificate \eqref{eq18}, and define 
    \begin{align}
        O^{\top}\bar z_T =: \begin{bmatrix}
            \bar a_T \\
            \bar b_T
        \end{bmatrix}, \qquad \bar a_T\in\mathbb R^n, \quad \bar b_T\in\mathbb R^p.
    \end{align}
    Multiplying \eqref{eq36} on the left by $\bar a_T^\top$ and subtracting $\bar b_T^\top w_T(u)$ from both sides of the equality we obtain 
    \begin{align}\label{eq:aT&bT}
        \bar a_T^{\top} y -\bar b_T^{\top} w_T(u) =& \left(\bm \zeta_T\bar a_T-\bar b_T\right)^{\top} w_T(u) \nonumber \\
        &+\bar a_T^{\top}\mathbf Q_T^{1/2}v+\bar a_T^{\top} e_T.
    \end{align}
    Define
    \begin{align}\label{eq:cT&etaT}
        \bar c_T &:= \bm \zeta_T\bar a_T-\bar b_T,\\
        \eta_T &:= \|\mathbf Q_T^{1/2}\bar a_T\|\; ,\nonumber
    \end{align}
    and define for a fixed $s\in\{+1,-1\}$, the following SOCP:
    \begin{align}
        \begin{aligned}
            \max_{u,t,\mu}\quad&\mu\\
            \text{s.t.}\quad
            &u\in\mathcal U,\\
            &\|\mathbf A_T^{-1/2}w_T(u)\|\leq t,\\
            &s~\bar c_T^{\top} w_T(u) -\eta_T~t -\sqrt{\bar a_T^{\top}\Omega_T\bar a_T} \geq\mu\; .
        \end{aligned}
        \tag{SOCP} \label{eq46}
    \end{align}
    
    To compute a weakest direction we solve \eqref{eq46} for $s=+1$ and $s=-1$ separately; we then choose any problem corresponding to the larger optimal value $\mu$; and we denote the corresponding input value by $u_T^{\rm fd}$.
    \begin{remark}%\rm 
        Note that the solution of the problem \ref{eq46} need not be unique, and consequently there may be more than one optimal input.\hfill $\blacksquare$
    \end{remark}
    \medskip 
    
    In the next result we show that  a suitable choice of the input provides a lower bound on the contribution of the new data in the weakest direction; moreover, we show that under suitable conditions the given certificate margin improves in the weakest direction. 
    \begin{theorem} \label{theorem1}
        Suppose that \eqref{eq40} holds, $\bar\alpha_T>0$, and $\bar z_T$ is selected according to \eqref{eq42}. If \eqref{eq46} has a solution $(u_T^{\rm fd},t_T,\mu_T)$ with $\mu_T\geq0$ for some $s\in\{+1,-1\}$, then for every $\Theta\in\Sigma_T$ and every future disturbance satisfying $e_Te_T^{\top}\preceq\Omega_T$ the following inequality holds: 
        \begin{align*}
            \left| \bar a_T^{\top} \bigl(\Theta w(u_T^{\rm fd})+e_T\bigr) -\bar b_T^{\top} w(u_T^{\rm fd}) \right| \geq\mu_T\; .     
        \end{align*}
        Moreover, if
        \begin{align*}
            \mu_T\ge \sqrt{\bar a_T^{\top}\Omega_T\bar a_T},
        \end{align*}
        then    
        \begin{align}\label{eq:weakimprovement}
            \bar z_T^{\top} \mathcal C_{T+1}(\bar\vartheta_T,\bar\alpha_T,\bar\rho_T) \bar z_T\ge \bar z_T^{\top} \mathcal C_T(\bar\vartheta_T,\bar\alpha_T,\bar\rho_T) \bar z_T\; .
        \end{align}
    \end{theorem}
    
    \begin{proof}
        Fix $u\in\mathcal U$ and let $y$ be a corresponding admissible next value of the state. By Lemma~\ref{lemma2}, there exist $v,e\in\R^n$ such that
        \begin{align*}
            y=\bar y_T(u)+\mathbf Q_T^{1/2}v+e,\quad \norm{v}_2\leq\sigma_T(u),\quad ee^\top\preceq\Omega_T.
        \end{align*}
        For each $s\in\{-1,1\}$, from the positions \eqref{eq:cT&etaT} and \eqref{eq:aT&bT} we conclude that
        \begin{align*}
            s\left(\bar a_T^\top y-\bar b_T^\top w_T(u)\right)&=s\bar c_T^\top w_T(u)+s\bar a_T^\top\mathbf Q_T^{1/2}v+s\bar a_T^\top e\\
            &\geq s\bar c_T^\top w_T(u)-\eta_T\sigma_T(u)-\sqrt{\bar a_T^\top\Omega_T\bar a_T}\; ,
        \end{align*}
        since $(\bar a_T^\top e)^2\leq\bar a_T^\top\Omega_T\bar a_T$ given that $ee^\top\preceq\Omega_T$. Evaluating this bound at the selected optimizer of \eqref{eq46} and combining with  $\mu_T\geq0$ proves the first inequality.
        
        To prove the rest of the claim recall the definitions \eqref{eq:h&R} and write 
        \begin{align*}
            \bar z_T^\top h_T(u,y)&=\bar a_T^\top y-\bar b_T^\top w_T(u)\\
            \bar z_T^\top R_T\bar z_T&=\bar a_T^\top\Omega_T\bar a_T\; .
        \end{align*}
        Applying Lemma~\ref{lemma1} to the reference solution gives
        \begin{align*}
            &\bar z_T^\top\left(\mathcal C_{T+1}(\bar\vartheta_T,\bar\alpha_T,\bar \rho_T)-\mathcal C_T(\bar\vartheta_T,\bar\alpha_T,\bar \rho_T)\right)\bar z_T\\
            =&\bar\alpha_T\left(\left|\bar z_T^\top h_T(u_T^{\rm fd},y)\right|^2-\mu^2_T+\mu^2_T-\bar a_T^\top\Omega_T\bar a_T\right)\; .
        \end{align*}
        Now use the majorization on $\mu_T$ established in the first part of the proof; the assumption $ \mu_T\ge \sqrt{\bar a_T^{\top}\Omega_T\bar a_T}$;  and the fact that the certificate is feasible and consequently that $\bar\alpha_T\geq0$, to conclude that
        \[
        \bar\alpha_T\left(\left|\bar z_T^\top h_T(u_T^{\rm fd},y)\right|^2-\bar a_T^\top\Omega_T\bar a_T\right)\ge0 \; ,
        \]
        and consequently that \eqref{eq:weakimprovement} is satisfied.
    \end{proof}
    \medskip
    
    Theorem \ref{theorem1} only guarantees improvement \emph{in the weakest direction} corresponding to a fixed reference certificate; it does \emph{not} imply that there is a corresponding improvement \emph{in the performance inde}x, i.e. that $\rho_{T+1}^\star\ge\rho_T^\star$. For example, consider the  case when $\mathcal{V}_T\subseteq\ker O^{\top}$; in such case $\bar a_T$ and $\bar b_T$ are zero, and \eqref{eq46} yields $\mu_T=0$. 
    % Also, if the current performance is mainly limited by a data-independent side constraint, improving a direction of the main criterion may still fail to improve the optimal value. 
    To guarantee that the overall performance does not decrease, or that it strictly increases, we  need the condition derived  in the next subsection.

    \subsection{Unified SDP}\label{sec:SDP}
    To prove the main result of this section we need three auxiliary results. We first define the  output-injection matrix
    \begin{align}
        J_N:=\begin{bmatrix}I_n\\0_{p\times n}\end{bmatrix}, \qquad J:=OJ_N.
    \end{align}
    Using the characterization \eqref{eq36} obtained in Lemma \ref{lemma2} we write  $h_T(u,y)=h_0(u)+\Delta h$, 
    where
    \begin{align}\label{eq:h0&Deltah}
        h_0(u):=O\begin{bmatrix}\bar y_T(u)\\-w_T(u)\end{bmatrix}, \quad 
        \Delta h:=J\bigl(\mathbf Q_T^{1/2}v+e_T\bigr)\; .
    \end{align}
    Define
    \begin{eqnarray}\label{eq:Qh&Rh}
        \mathbf Q_h&:=&J\mathbf Q_TJ^{\top}\succeq0,\nonumber\\
        \mathbf R_h&:=&J\Omega_TJ^{\top}\succeq0
    \end{eqnarray}
    and define $L_T:=J\Omega_T^{1/2}$; note that $\mathbf R_h=L_TL_T^{\top}$.

    Our first auxiliary result concerns an upper bound on the matrix $\Delta h \Delta h^\top$. 
    \begin{lemma} \label{lemma3}
        Define $\mathbf Q_h$ and $\mathbf R_h$ by \eqref{eq:Qh&Rh}, and $\sigma_T(u)$ as in Lemma \ref{lemma2}.  The following inequality holds: 
        \begin{align}
            \Delta h\Delta h^{\top} \preceq 2\sigma_T^2(u)\mathbf Q_h+2\mathbf R_h.
        \end{align}
    \end{lemma}
    
    \begin{proof}
        Recall that for any pair of vectors $a$, $b$ it holds that 
        \[
        (a+b)(a+b)^\top \preceq 2 (a a^\top +b b^\top)\; .
        \]
        Use such inequality and the definition of $\Delta h$ to conclude that 
        \[
        \Delta h\Delta h^{\top} \preceq 2 J \left(\mathbf Q_T^\frac{1}{2}v v^\top \mathbf Q_T^\frac{1}{2}+ e e^\top\right)J^\top \; .
        \]
        Since $v v^\top \preceq \sigma_T(u)^2 I$ and $e e^\top \preceq \Omega_T$, it follows  that 
        \[
        \Delta h\Delta h^{\top} \preceq 2 \sigma_T(u)^2 \mathbf Q_h+ 2 \mathbf R_h\; ;
        \]
        this concludes the proof. 
    \end{proof}
    \medskip
    
    We now define a \emph{cover} of the set of possible  errors $y-\bar y_T(u)$ corresponding to the ``prediction equation" \eqref{eq36}. We construct such cover with  smaller, possibly overlapping, ellipsoids. Later in this section we use such sets to compute a bound on the performance improvement on each of such sub-regions given, eventually selecting for input design the sub-region corresponding with the best such bound.

    Recall the definition of $\sigma_T(u)$ in Lemma \ref{lemma2}, and choose a finite bound $\bar\sigma_T$ satisfying
    \begin{align*}
        \sigma_T(u)\le\bar\sigma_T\qquad\text{for every }u\in\mathcal U.
    \end{align*}
    Note that such bound exists since $\mathcal U$ is compact, see Assumption \ref{ass:compact}. For example, if $\mathcal{U}$ is defined by the 2-norm bound $\|u\|_2\le\bar u_2$, using the triangle inequality we conclude that $\bar\sigma_T$ can be chosen as 
    \begin{align*}
        \bar\sigma_T:={}&\left\|\mathbf A_T^{-1/2}
        \begin{bmatrix}x_T\\0\end{bmatrix}\right\|_2
        +\bar u_2\left\|\mathbf A_T^{-1/2}
        \begin{bmatrix}0\\I_m\end{bmatrix}\right\|_2.
    \end{align*}
    By Lemma~\ref{lemma2}, the  ``prediction error" $d:=y-\bar y_T(u)$ between the observed  and the future state corresponding to $u_T$, $x_T$ and $e_T$ belongs to the compact set
    \begin{align*}
        \mathcal R_T:=\{\mathbf Q_T^{1/2}v+\Omega_T^{1/2}z: \|v\|_2\le\bar\sigma_T,\ \|z\|_2\le1\}, 
    \end{align*}
    the sum of two ellipsoids. 
    
    %Before solving the input design problem, 
    Now fix centers $c_j\in\mathbb R^n$ and matrices $D_j=D_j^\top\succ0$, $j=1,\ldots,M_T$, such that
        \begin{align}
            \mathcal R_T\subseteq\bigcup_{j=1}^{M_T}\mathcal R_{T,j},
            \qquad
            \mathcal R_{T,j}:=\{c_j+D_jz:\|z\|_2\le1\}. \label{eq:cover}
        \end{align}
        It can be shown that a finite cover always exists; for instance, a bounded box containing $\mathcal R_T$ can be covered by finitely many Euclidean balls (see Appendix~\ref{app:cover} for a construction of such type). The centers and matrices are fixed coefficients of the SDP. For each region define
        \begin{align}
            h_{0,j}(u)&:=h_0(u)+J_Nc_j,\nonumber\\
            B_j&:=JD_j,\nonumber\\
            \bar h_j&:=h_{0,j}(u_T^{\rm fd}).\label{eq:localreference}
    \end{align}
    
    \begin{lemma}\label{lemma4add}
        If $d=y-\bar y_T(u)\in\mathcal R_{T,j}$  for some $j\in\{1,\ldots,M_T\}$, then 
        there exists $z$ such that $\|z\|_2\le1$ and  
        \begin{align}
            h_T(u,y)=h_{0,j}(u)+\Delta h_j,
        \end{align}
        where $\Delta h_j:=B_jz$. Moreover, $\Delta h_j\Delta h_j^\top\preceq B_jB_j^\top$.
    \end{lemma}
    \begin{proof}
        By the definition of $\mathcal R_{T,j}$, for every $d\in\mathcal{R}_{T,}j$ there exists $z$ with
        $\|z\|_2\le1$ such that $d=c_j+D_jz$. Hence
        \begin{align*}
            h_T(u,y)
            &=h_0(u)+Jd\\
            &=h_0(u)+Jc_j+JD_jz
            =h_{0,j}(u)+B_jz.
        \end{align*}
        Moreover, $\|z\|_2\le1$ implies $zz^\top\preceq I_n$, and therefore
        \begin{align*}
            \Delta h_j\Delta h_j^\top
            =B_jzz^\top B_j^\top
            \preceq B_jB_j^\top.
        \end{align*}
        The claim is proved.
    \end{proof}
        \medskip
    
    Now define the scaled set
    \begin{align}
        \widehat{\mathcal U} :=\left\{ (\alpha,\widehat u):\alpha>0,\ \widehat u/\alpha\in\mathcal U \right\}\; ,\label{eq:hatU}
    \end{align}
    and define
    \begin{align}
        \widehat u:=\alpha u, \qquad \widehat w:=\alpha~ w_T(u) =\begin{bmatrix}\alpha~ x_T\\\widehat u\end{bmatrix}\; .
    \end{align}
    Note that the constraint $(\alpha,\widehat u)\in\widehat{\mathcal U}$ is convex. If
    $\alpha>0$, the value of the physical input $u$ can be recovered as $u=\widehat u/\alpha$.  Recall the definition of $h_0$ in \eqref{eq:h0&Deltah} and define 
    \begin{align}\label{eq:wh0&Deltah}
        \widehat h_0 :=\alpha h_0(u) =O\begin{bmatrix}
            \bm \zeta_T^{\top}\widehat w \\
            -\widehat w
        \end{bmatrix}\; ;
    \end{align}
    note that $\widehat h_0$ is affine in $(\alpha,\widehat u)$.
    
    \begin{remark}[Perspective constraints]\rm 
        If 
        \[
        \mathcal U=\{u:\|u\|_\infty\le\bar u\}\; ,
        \]
        then the  perspective constraint (constraint on the scaled set) is
        $|\widehat u_i|\le\alpha\bar u_i$. If
        \[
        \mathcal U=\{u:(u-u_c)^{\top} H_u^{-1}(u-u_c)\le1\}\; , 
        \]
        then the constraint is $
        \|H_u^{-1/2}(\widehat u-\alpha u_c)\|_2\le\alpha$. \hfill $\blacksquare$
    \end{remark}
    \medskip
    
    We now prove another instrumental result. For each region in \eqref{eq:cover}, we first fix a reference vector $\bar h_j$ for the convex subproblem with $T$ data samples; for example, $\bar h_j$ can be  computed from the SOCP input $u_T^{\rm fd}$ of Theorem \ref{theorem1} as $\bar h_j=h_{0,j}\left(u_T^{\rm fd}\right)$.
    
    %\begin{lemma} \label{lemma4}
    %For any $\theta>0$ and any $h_T=h_T(u,y)$, the following inequality holds: 
    %\begin{align}
    %\alpha h_T h_T^{\top} \succeq \operatorname{He}(\bar h_T\widehat h_0^{\top}) -(\alpha+\theta)\bar h_T\bar h_T^{\top} -\frac{\alpha^2}{\theta}\Delta h\Delta h^{\top}.
    %\end{align}
    %\end{lemma}
    
    % \begin{proof}
        % From $\alpha(h_T-\bar h_T)(h_T-\bar h_T)^\top\succeq0$ we conclude that 
        % \begin{align*}
            %     \alpha~ h_T h_T^\top&\succeq\He\!\left(\alpha~\bar h_Th_T^\top\right)-\alpha~\bar h_T\bar h_T^\top\\
            %     &=\He\!\left(\bar h_T\widehat h_{0,T}^\top\right)+\He\!\left(\alpha\bar h_T\Delta h_T^\top\right)-\alpha~\bar h_T\bar h_T^\top.
            % \end{align*}
        % Using the matrix Young inequality we conclude that for every $\theta>0$ the inequality 
        % \begin{align*}
            %     \He\!\left(\alpha\bar h_T\Delta h_T^\top\right)\succeq-\theta\bar h_T\bar h_T^\top-\frac{\alpha^2}{\theta}\Delta h_T\Delta h_T^\top \; ,
            % \end{align*}
        % holds. Consequently 
        % \begin{align}
            %     \alpha h_T(u,y)h_T(u,y)^\top\succeq{}&\He\!\left(\bar h_T\widehat h_{0,T}^\top\right)-(\alpha+\theta)\bar h_T\bar h_T^\top\notag\\
            %     &-\frac{\alpha^2}{\theta}\Delta h_T\Delta h_T^\top\; ;\label{theorem2.2}
            % \end{align}
        % the claim is proved. 
        % \end{proof}
    % \medskip 
    
    For each region in \eqref{eq:cover}, the scaled center is
    \begin{align}
        \widehat h_{0,j}:=\widehat h_0+\alpha Jc_j =\alpha h_{0,j}(u), \label{eq:scaledlocal}
    \end{align}
    which is affine in $(\alpha,\widehat u)$.
    \begin{lemma}\label{lemma4}
        Assume that $d=y-\bar y_T(u)\in\mathcal R_{T,j}$ for some $j\in\{1,\ldots,M_T\}$. The following inequality holds for every $\theta_j>0$: 
        \begin{align}
            \alpha hh^\top\succeq{} \operatorname{He}(\bar h_j\widehat h_{0,j}^\top) -(\alpha+\theta_j)\bar h_j\bar h_j^\top-\frac{\alpha^2}{\theta_j}B_jB_j^\top. \label{eq62}
        \end{align}
    \end{lemma}
    \begin{proof}
        Since $\alpha(h-\bar h_j)(h-\bar h_j)^\top\succeq0$,
        \begin{align*}
            \alpha hh^\top\succeq{}& \operatorname{He}(\bar h_j\widehat h_{0,j}^\top) -\alpha\bar h_j\bar h_j^\top +\operatorname{He}(\alpha\bar h_j\Delta h_j^\top).
        \end{align*}
        Young's matrix inequality gives
        \begin{align*}
            \operatorname{He}(\alpha\bar h_j\Delta h_j^\top) \succeq-\theta_j\bar h_j\bar h_j^\top-\frac{\alpha^2}{\theta_j}\Delta h_j\Delta h_j^\top.
        \end{align*}
        Now use the majorization $\Delta h_j\Delta h_j^\top\preceq B_jB_j^\top$ to conclude that \eqref{eq62} holds.
    \end{proof}

    Using Lemmas \ref{lemma3} -- \ref{lemma4} we are now in a position to prove the main result of this section. 
    % We first choose $\kappa$ such that 
    % \begin{align}
        %     \begin{bmatrix}
            %         \kappa & \widehat w^{\top} \\
            %         \widehat w & \theta\mathbf A_T
            %     \end{bmatrix}\succeq0.
        % \end{align}
    % Since $\mathbf A_T\succ0$, a Schur complement argument shows that 
    % \begin{align}
        %     \kappa \ge \frac{\alpha^2}{\theta} w_T(u)^{\top}\mathbf A_T^{-1}w_T(u) = \frac{\alpha^2}{\theta}\sigma^2(u).
        % \end{align}
    % Combining Lemma \ref{lemma3} and Lemma \ref{lemma4} gives
    % \begin{align}
        % \alpha hh^{\top} \succeq \operatorname{He}(\bar h\widehat h_0^{\top}) -(\alpha+\theta)\bar h\bar h^{\top} -2\kappa\mathbf Q_h -\frac{2\alpha^2}{\theta}\mathbf R_h. \label{eq62}
        % \end{align}
    %Now fix a performance level $\bar\rho_T$ \blue{for $T$  data samples} and a reference vector $\bar h_T$. Define
    % \begin{align}
        % X_T := &\mathcal C_T(\vartheta,\alpha,\bar\rho_T+\delta) +\operatorname{He}(\bar h\widehat h_0^{\top}) -(\alpha+\theta)\bar h\bar h^{\top} \notag\\
        % &-\alpha\mathbf R_h -2\kappa\mathbf Q_h.
        % \end{align}

        Fix  a performance level $\bar\rho_T$, a constant $\nu_T\ge0$, and define the reference vectors in \eqref{eq:localreference}. Define
        \begin{align}
            X_{T,j}:={}&
            \mathcal C_T(\vartheta_j,\alpha,\bar\rho_T+\delta)+\operatorname{He}(\bar h_j\widehat h_{0,j}^\top)\notag\\
            &-(\alpha+\theta_j)\bar h_j\bar h_j^\top-\alpha\mathbf R_h-\nu_T I.\label{eq:localX}
        \end{align}

    Now  define the following optimization problem: 
    % \begin{align}
        % \begin{aligned}
            %     \max_{\vartheta,\alpha,\widehat u,\delta,\theta,\kappa}
            %     \quad&\delta\\
            %     \text{s.t.}\quad
            %     &(\vartheta,\alpha,\bar\rho_T+\delta)\in\mathcal K,\\
            %     &(\alpha,\widehat u)\in\widehat{\mathcal U},\\
            %     &\theta>0,\\
            %     &\begin{bmatrix}
                %         \kappa & \widehat w^{\top} \\
                %         \widehat w & \theta\mathbf A_T
                %     \end{bmatrix}\succeq0,\\
            %     &\begin{bmatrix}
                %         X_T & \sqrt2\,\alpha L_T \\
                %         \sqrt2\,\alpha L_T^{\top} & \theta I_n
                %     \end{bmatrix}\succeq0\; .
            % \end{aligned}
        % \tag{SDP} \label{eq64}
        % \end{align}
    % Since $\widehat h_0$ is affine in $(\alpha,\widehat u)$ and $\mathcal C_T$ is affine in
    % $(\vartheta,\alpha,\rho)$, the problem \eqref{eq64} is a semi-definite program.

        \begin{align}
            \begin{aligned}
                \max_{\alpha,\widehat u,\delta,\{\vartheta_j,\theta_j\}_{j=1}^{M_T}}\quad&\delta\\
                \text{s.t.}\quad
                &(\alpha,\widehat u)\in\widehat{\mathcal U},\\
                &(\vartheta_j,\alpha,\bar\rho_T+\delta)\in\mathcal K,\quad j=1,\ldots,M_T,\\
                &\theta_j>0,\quad j=1,\ldots,M_T,\\
                &\delta\ge\varepsilon_T,\\
                &\begin{bmatrix}
                    X_{T,j}&\alpha B_j\\
                    \alpha B_j^\top&\theta_j I_n
                \end{bmatrix}\succeq0,
                \quad j=1,\ldots,M_T.
            \end{aligned}
            \tag{SDP}\label{eq64}
        \end{align}
        It is a semi-definite program: all matrix entries are affine in the decision variables and   $\alpha$ is a common multiplier.
    
    \begin{theorem}\label{theorem2}
        Suppose that \eqref{eq64} has a feasible solution
        $(\alpha_T,\widehat u_T,\delta_T, \{\vartheta_j,\theta_{T,j}\}_{j=1}^{M_T})$.
        If $\alpha_T>0$, applying $u_T=\widehat u_T/\alpha_T$, then for every admissible future observation the  optimal value $\rho_{T+1}^\star$ obtained by resolving the original SDP with the actual new data satisfies 
        \begin{align}
            {\rho_{T+1}^\star\ge\bar\rho_T+\delta_T.} \label{eq65}
        \end{align}
        Moreover, if  
            \begin{align*}
                \delta_T\ge \varepsilon_T,
            \end{align*}
            where $\varepsilon_T$ is given in \eqref{eq:barrhoT}, then 
            \begin{align}
                \rho_{T+1}^\star\ge \rho_T^\star. \label{strongimprovement}
        \end{align}
    \end{theorem}
    
    % \red{{\bf PAOLO:} I think we can drop $(y)$ in $\rho^\star_T\red{(y)}$. The quantity $\rho^\star_T\red{(y)}$ is not defined; we have only defined $\rho^\star_T$. Sorry for the sloppy notation in my previous note. The red \red{$(\ldots)$} marks parts that can be deleted.}
    
    \begin{proof}
        % Take any admissible future observation. Applying the Schur complement to the LMI
        % \[
        % \begin{bmatrix}
            %         X_T & \sqrt2\,\alpha L_T \\
            %         \sqrt2\,\alpha L_T^{\top} & \theta I_n
            %     \end{bmatrix}\succeq0
        %     \]
        %     gives
        % \begin{align}
            %     X_T-\frac{2\alpha_T^2}{\theta_T}\mathbf R_h\succeq0.
            % \end{align}
        % By \eqref{eq62} and \eqref{eq33},
        % \begin{align}
            % \mathcal C_{T+1}(\vartheta_T,\alpha_T,\bar\rho_T+\delta_T)
            % =&\mathcal C_T(\vartheta_T,\alpha_T,\bar\rho_T+\delta_T)
            % +\alpha_T hh^{\top}\notag\\
            % &-\alpha_T\mathbf R_h\notag\\
            % \succeq& X_T-\frac{2\alpha_T^2}{\theta_T}\mathbf R_h
            % \succeq0.
            % \end{align}

            By Lemma~\ref{lemma2} and the definition \eqref{eq:cover} of the cover, the observed prediction error $d=y-\bar y_T(u)$ belongs to at least one region $\mathcal R_{T,j}$. For such region, the Schur complement of $\theta_j I_n\succ 0$ in the last constraint in \eqref{eq64} gives
            \begin{align*}
                X_{T,j}-\frac{\alpha_T^2}{\theta_{T,j}}B_jB_j^\top\succeq0.
            \end{align*}
            Recall the definition \eqref{eq:localX} of $X_{T,j}$, and apply Lemmas~\ref{lemma1} and~\ref{lemma4}; conclude that 
            \begin{align*}
                &\mathcal C_{T+1}(\vartheta_{T,j},\alpha_T,\bar\rho_T+\delta_T)\\
                &\quad\succeq X_{T,j}-\frac{\alpha_T^2}{\theta_{T,j}}B_jB_j^\top+\nu_T I\succeq\nu_T I\succeq 0.
            \end{align*}
            The same region's variables satisfy all constraints in $\mathcal K$. Hence $\bar\rho_T+\delta_T$ is feasible for the updated certificate problem, proving \eqref{eq65}. Combining this with $\delta_T\ge \varepsilon_T$ and  \eqref{eq:barrhoT} proves \eqref{strongimprovement}. 

        %The same pair $(\vartheta_T,\alpha_T)$ still satisfies all data-independent side constraints. Therefore, $\bar\rho_T+\delta_T$ is a feasible performance level for the updated original SDP, and optimality gives \eqref{eq65}. \blue{Combining the assumption $\delta_T\ge\varepsilon_T$ with \eqref{eq:barrhoT} yields \eqref{strongimprovement}.}
    \end{proof}
    \medskip
    \begin{remark}\rm The result of Theorem \ref{theorem2} depends on the choice of the constant $\nu_T$ in the definition of the regions $X_{T,j}$, see \eqref{eq:localX}. If $\nu_T>0$ then a strict matrix inequality results for $\mathcal C_{T+1}(\vartheta_{T,j},\alpha_T,\bar\rho_T+\delta_T)$, which may be useful for some applications.\hfill $\blacksquare$
    \end{remark}
    \medskip
    
    % The region, and therefore the certificate, is selected after observing $y$; the applied input is the same for every region.
    
    \section{Input-design procedure and applications}\label{sec:algo}
    We formulate an algorithm for input design based on  Theorem \ref{theorem1} and Theorem \ref{theorem2}, and we illustrate its application to the three data-driven control problems considered in Examples \ref{ex:QS}--\ref{ex:Diss}. \medskip

    % For the selected task, let $\rho_T^\star$ denote the value defined in the
    % corresponding example.  Choose a currently feasible baseline
    % $\bar\rho_T$.  If the supremum is not known to be attained, then, for any
    % prescribed $\varepsilon_T>0$, it can be chosen so that
    % \begin{align}
        %     \bar\rho_T\geq \rho_T^\star-\varepsilon_T.
        % \end{align}
    % If the optimum is attained, one may take
    % $\varepsilon_T=0$ and $\bar\rho_T=\rho_T^\star$.
    
    \subsection{Input-design procedure}
    
    We now combine \eqref{eq46} and \eqref{eq64} in an online input design procedure. The SOCP is used to construct the reference vector, and the SDP selects the next input subject to $\delta\ge\varepsilon_T$, ensuring that the optimal performance value does not decrease.\medskip
    
    \textbf{Input-design procedure}\medskip
    
    \begin{description}
        \item[{\bf Input:}]~~Initial data record $\mathcal D_{T_{\rm ini}}$, admissible input set $\mathcal U$, terminal time
            $T_{\max}>T_{\rm ini}$, and strictly positive tolerances
            $\{\varepsilon_t\}_{t=T_{\rm ini},\ldots,T_{\max}-1}$.
        \item[{\bf Output: }]~~ A feasible input sequence $\{u_t\}_{t=T_{ini},\ldots,T_{\max-1}}$ such that $\rho^\star_{T+1}\ge\rho^\star_T$. \medskip 
        
        \item[Step 1:]~~ $T:=T_{ini}$. \smallskip
        \item[Step 2:]~~ If $T\ge T_{\max}$, terminate. \smallskip
        \item[Step 3:]~~ Compute $\bar \rho_T$ for the prescribed $\varepsilon_T>0$; compute feasible reference solution $(\bar\vartheta_T,\bar\alpha_T,\bar\rho_T)$  satisfying \eqref{eq:barrhoT} and \eqref{eq40} with $\bar\alpha_T>0$ (see Remark \ref{rem:brhoT}).\smallskip
        \item[Step 4:]~~ Compute $N_T$ (from \eqref{eq13}); $\mathbf A_T$ (from \eqref{eq:ssellipsoid});  $\bm \zeta_T$ and $\mathbf Q_T$ (from \eqref{eq:zeta&QT}); $\mathbf Q_h$ and $\mathbf R_h$ (from \eqref{eq:Qh&Rh}).\smallskip
        \item[Step 5:]~~ Solve \eqref{eq46} for $s=1$ and $s=-1$. Choose an optimal solution attaining the larger optimal value, denote its input by $u_T^{\rm fd}$. 
            \item[Step 6:]~~For every region compute $\bar h_j=h_{0,j}(u_T^{\rm fd})$ and $B_j=JD_j$ (from \eqref{eq:localreference}). \smallskip
        \item[Step 7:]~~Solve \eqref{eq64}. If the problem is infeasible, then apply $u_T=u_T^{\rm fd}$ and go to Step 9.\smallskip
        \item[Step 8:]~~Apply $u_T=\widehat u_T/\alpha_T$ using the solution obtained in Step 7.\smallskip
        % \item[\textcolor{red}{Step 7:}]~~ \textcolor{red}{What if $\alpha_T=0$?} \smallskip
        \item[Step 9:]~~Measure $y=x_{T+1}$, update the data record
            $\mathcal D_{T+1}$. update $N_{T+1}$ using \eqref{eq34}.\smallskip
        %\item[\red{Step 10:}]~~\red{Compute the value of $\bar \rho_{T+1}$ for $\mathcal C_{T+1}(\vartheta,\alpha,\rho)$.}\smallskip
        \item[Step 10:]~~$T:=T+1$; go to Step 2. 
    \end{description}
    \bigskip
    
    In Step 3 we compute the performance at time $T$; this step can result in the optimal performance ($\rho_T=\rho_T^\star$) or a strictly feasible baseline value $\bar\rho_T$ for it. In Steps 4 and 5, we  compute a weakest direction as a reference vector. In Step 6, we compute the cover of the error space; in Step 7, we solve the one-step perspective problem. If such problem is feasible  with value $\delta_T$, then for every admissible future observation $y$ corresponding to the input computed in Step 8 it holds that 
    \begin{align*}
        \rho_{T+1}^\star(y) \geq \bar\rho_T+\delta_T \geq \rho_T^\star+\delta_T-\varepsilon_T\; .
    \end{align*}
    Note that $\delta_T\geq\varepsilon_T$ ensures that the value of the 
    $T_{\max}>T_{\rm ini}$, and strictly positive tolerances
    optimal performance does not decrease; while $\delta_T>\varepsilon_T$ guarantees a strict increase of performance.  We continue to apply the steps until we reach the prescribed number of data samples $T_{max}$. \medskip 
    
    % \textcolor{red}{{\bf YISHU:} I have  left out the formal proof of the correctness. It seemed excessive to me. A question- what is the termination condition? We need to have one, of course. Maybe your  work on the code and the examples will suggest one (maybe "if performance does not improve more than...stop"?}
    
    % \textcolor{blue}{\textbf{PAOLO:} Your revisions to the proof are excellent! I have added termination conditions to the algorithm, with the text to be deleted marked in red. I have also made a few small changes to the algorithm (apologies for the inaccuracies in my earlier note). In fact, by the definition in \eqref{eq:hatU}, the constraint $(\alpha,\widehat u)\in\widehat{\mathcal U}$ ensures that every feasible solution satisfies $\alpha_T>0$, so the case $\alpha_T=0$ does not arise. }  
    
    \subsection{Applications to three control problems}
    We apply the procedure to the three control problems introduced in Examples~\ref{ex:QS}--\ref{ex:Diss}.
    
    \subsubsection{Quadratic stabilization}\label{sec:algoQS}
    %\begin{example}[{Quadratic stabilization}]\rm\label{sec:algoQS}
    We denote by $\rho_T^\star$ the normalized stabilization-certificate
    margin defined in Example~\ref{ex:QS}.  Substituting
    $(\vartheta,\mathcal F,O,\mathcal K)=( (P,L),\mathcal F_{\rm stab},O_{\rm stab},\mathcal K_{\rm stab})$
    into the unified design gives
    \begin{align}
        \rho_{T+1}^\star \geq \bar\rho_T+\delta_T\ge \rho_T^\star.
    \end{align}
    If there exists $T_1$ such that 
    \begin{align}
        \bar\rho_{T_1}+\delta_{T_1}>0,
    \end{align}
    then the SDP certificate itself is a strict common quadratic stabilization certificate for all systems consistent with the updated data.  In particular, the corresponding $P$ is positive definite and the feedback gain is recovered as
    \begin{align}
        K=LP^{-1}.
    \end{align}
    %\end{example}
    %\medskip 
    
    \subsubsection{$\mathcal H_2$ control}\label{sec:algoH2}
    %\begin{example}[$\mathcal H_2$ control]
    In Example~\ref{ex:H2}, $\rho=-\gamma_2^2$.  Define the common
    $\mathcal H_2$ bound certified by the data as
    \begin{align}
        \gamma_{2,T}:=\sqrt{-\rho_T^\star}.
    \end{align}
    It follows from \eqref{eq65} and \eqref{eq:barrhoT} that 
    \begin{align}
        \bigl(\gamma_{2,T+1}\bigr)^2&\le -\bar\rho_T-\delta_T \nonumber\\
        &\le \bigl(\gamma_{2,T}\bigr)^2 +\varepsilon_T-\delta_T.
    \end{align}
    Thus, $\gamma_{2,T}$ does not increase if $\delta_T\ge\varepsilon_T$, and strictly decreases if $\delta_T>\varepsilon_T$. Since $\bar\rho_T+\delta_T\le0$, every performance level
    \begin{align}
        \gamma_2>\sqrt{-\bar\rho_T-\delta_T}
    \end{align}
    is certifiable for all systems consistent with the updated data.
    %\end{example}
    %\medskip 

    \subsubsection{Strict dissipativity}\label{sec:algoSD}
    %\begin{example}[{Strict dissipativity}]\rm\label{sec:algoSD}
    In Example~\ref{ex:Diss}, $\rho_T^\star$ is the uniform strictness margin of the dualized dissipativity certificate associated with the fixed supply matrix $S$. If
    \begin{align}
        \bar\rho_T+\delta_T>0,
    \end{align}
    then the solution of \eqref{eq64} provides a common quadratic storage
    function establishing strict dissipativity for all systems consistent
    with the updated data, with respect to the same fixed supply rate.
    %\hfill $\blacksquare$
    %\end{example}
    \medskip 
    
    In Appendix \ref{app:3morappl} we illustrate the application of the procedure to the additive and multiplicative feedback fragility, and the $\mathcal{H}_\infty$-control problems. 
    
    \section{Numerical examples}\label{sec:numerical}
    
    In Section \ref{sec:onlVSoffl} we  examine online input design for three control objectives. In Section \ref{sec:offlVSoffl} compare the number of samples required by \eqref{eq46} and \eqref{eq64} to attain a prescribed stabilization margin. Each comparison uses 100 independent trials, with the same initial data and disturbance sequence for all methods within each trial. Target attainment refers to the first update at which the prescribed level is reached. The accompanying MATLAB implementation uses YALMIP and MOSEK, with $\varepsilon_T=10^{-5}$.
    
    \subsection{Online vs. off-line input design for a two-state system}\label{sec:two-state system}\label{sec:onlVSoffl}
    In this section we show that the online procedures are more efficient than the off-line procedure both in terms of data samples and of achieved performance $\rho^\star$. 
    
    Consider the system in  Section IV of \cite{CoulsonQuantitativePE}
    \begin{align*}
        A_\star=\begin{bmatrix}1.5&1\\0&1\end{bmatrix},\qquad
        B_\star=\begin{bmatrix}0\\1\end{bmatrix}.
    \end{align*}
    For the quadratic stabilization problem (see Example \ref{ex:QS}), we choose $\mathcal{U}$ described by the constraint $|u_t|\leq0.21$; and the noise model with parameters $\Omega=0.02I_2$ (Assumption \ref{ass:indivbound})  and $\Delta_T=T\Omega$. The initial data are
    \begin{align*}
        W_3&=\begin{bmatrix}
            3.310&1.282&-1.803\\
            -3.632&-3.733&-3.480\\
            -0.168&0.168&-0.168
        \end{bmatrix}.
    \end{align*}
    The final state of the initial record is $x_3=[-6.249,\,-3.703]^\top$. 
    Future disturbances are uniform on the circle of radius $\sqrt{0.02}$. Note that although $W_3$ has full row rank, the initial margin $\rho_3^\star=-9.86\times10^{-4}$ is negative. Hence no controller can quadratically stabilize all systems consistent with the current data.
    
    First, we collect 70 additional samples using inputs drawn uniformly from $[-0.21,0.21]$. Figure~\ref{fig:random} shows the mean performance and the variation across trials. None of the 100 final records provides a positive stabilization certificate.
    
    \begin{figure}[!htbp]
        \centering
        \includegraphics[width=\columnwidth]{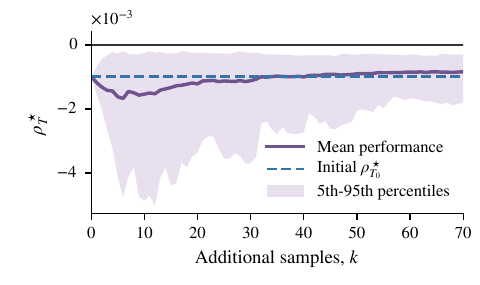}
        \caption{Random inputs over 70 updates. The solid curve is the mean
            over 100 trials; shading gives the 5th--95th percentiles. }
        \label{fig:random}
    \end{figure}
    
    Then, starting from the same data and using a cover with $M_T=856$ subregions, the online policies obtain a positive margin within ten additional samples in 76 trials, with a median of 7 samples. 
    A random input does so in only 4 trials within ten samples. The SDP uses $M_T$ covering regions and applies the SOCP input at every update in this experiment. Figure~\ref{fig:stab} shows that under both online policies the mean margin becomes positive after 9 updates.
    
    \begin{figure}[!htbp]
        \centering
        \includegraphics[width=\columnwidth]{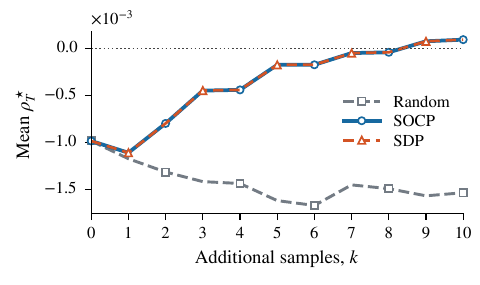}
        \caption{Quadratic stabilization: mean performance over 100 trials.}
        \label{fig:stab}
    \end{figure}
    
    For the $\mathcal H_2$ control problem (see Example \ref{ex:H2}), we retain $A_\star,B_\star$ as above and we set $C=\begin{bmatrix}1&0\end{bmatrix}$, $D=0$. The admissible inputs are constrained by $|u_t|\leq0.4$.
    We choose $\Delta_5=0.01I_2$, $\Omega=0.0008I_2$, and
    \begin{center}
        \resizebox{\columnwidth}{!}{$\displaystyle
            W_5=\begin{bmatrix}
                -0.070&-0.081&0.195&0.699&1.276\\
                0.030&0.312&0.419&0.234&-0.024\\
                0.270&0.120&-0.190&-0.270&0.270
            \end{bmatrix}.$}
    \end{center}
    The final state of the initial record is $x_5=[1.903,\,0.243]^\top$, and the initial data certify $\gamma_{2,5}^2=-\rho_5^\star=14.28$. Future disturbances are uniform in the disk of radius $\sqrt{0.0008}$, and $\Delta_{5+k}=\Delta_5+k\Omega$. $M_T = 224$. 
    
    Both online policies attain the target within ten samples in all 100 trials.
    %The SDP uses 224 covering regions. 
    Figure~\ref{fig:h2} shows that both online policies reduce the squared performance bound to $7$, a reduction of approximately 51\%, within ten samples in 98 trials. The median is 8 samples. Random input attains the same bound in only 3 trials. At $k=10$, the proportions are 98\% for both online policies and 3\% for random input.
    
    \begin{figure}[!htbp]
        \centering
        \includegraphics[width=\columnwidth]{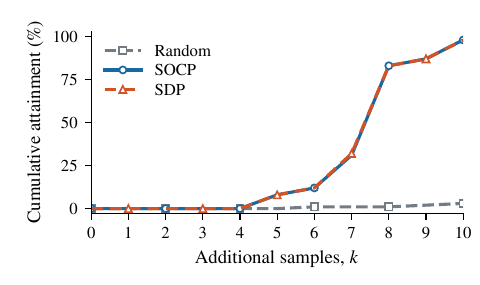}
        \caption{$\mathcal H_2$ control: proportion of trials attaining $\gamma_2^2\leq7$ within $k$ additional samples.}
        \label{fig:h2}
    \end{figure}
    
    For the  strict dissipativity problem (see Example \ref{ex:Diss}), we choose 
    \begin{align*}
        C=\begin{bmatrix}1.5&2\end{bmatrix},\qquad D=1,\qquad
        S=\operatorname{diag}(-0.025,1).
    \end{align*}
    The admissible inputs are constrained by $|u_t|\leq1$; the noise  has parameters $\Delta_3=0.1I_2$, $\Omega=0.002I_2$, and
    \begin{align*}
        W_3=\begin{bmatrix}
            -0.080&0.032&-0.006\\
            0.150&-0.051&-0.249\\
            -0.200&-0.200&-0.200
        \end{bmatrix},\qquad
        x_3=\begin{bmatrix}-0.257\\-0.446\end{bmatrix}.
    \end{align*}
    The disturbances have norm $0.8\sqrt{0.002}$ and uniformly distributed directions. We use $\Delta_{3+k}=\Delta_3+k\Omega$.
    The initial margin is $\rho_3^\star=-0.513$. The number of sub-regions for the cover is $M_T= 224$.  
    As shown in Figure~\ref{fig:diss}, both online policies attain $\rho_{\mathrm{tar}}=0.1$ in \emph{all} 100 trials within ten samples, with a median of 8 samples. Random input does not attain such level in \emph{any} trial.
    
    \begin{figure}[!htbp]
        \centering
        \includegraphics[width=\columnwidth]{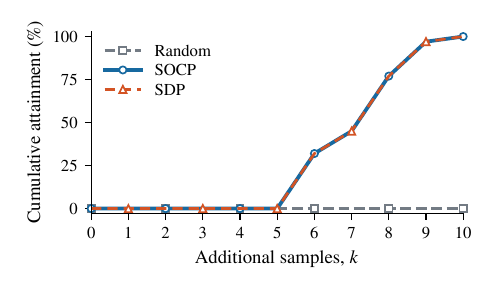}
        \caption{Strict dissipativity: proportion of trials attaining $\rho_{\mathrm{tar}}=0.1$ within $k$ additional samples.}
        \label{fig:diss}
    \end{figure}
    
    \subsection{Sample requirements of the SDP and SOCP policies}\label{sec:offlVSoffl}
    In the previous section we compared 
    the online policies SDP and SOCP with the off-line procedure. From the examples it was not apparent whether the additional computations involved in SDP were worth the improvement in performance, see Figures \ref{fig:stab}, \ref{fig:h2}, \ref{fig:diss}. We now compare the two online policies with each other on data generated by a three-state, two-input system described by
    \begin{align*}
        A_\star&=\begin{bmatrix}
            0.49023&-0.25557&-0.02215\\
            -0.13020&0.78215&0.20542\\
            0.05979&0.25045&0.78101
        \end{bmatrix},\\
        B_\star&=\begin{bmatrix}
            0.42538&-0.68301\\
            3.38485&1.76716\\
            0.08851&-1.72495
        \end{bmatrix},\qquad \|u_t\|_2\leq1.8.
    \end{align*}
    \begin{samepage}
        The noise has parameters $\Delta_5=0.01I_3$, $\Omega=0.0005I_3$, and $\Delta_{5+k}=\Delta_5+k\Omega$. The initial data is
        \begin{center}
            \resizebox{\columnwidth}{!}{$\displaystyle W_5=\begin{bmatrix}
                    0.650170&0.426918&1.401416&0.367489&0.153305\\
                    0.057460&-3.288858&-2.157021&-3.130513&-3.053396\\
                    -0.538980&0.537281&0.255565&-1.662076&-3.703085\\
                    -0.636170&0.290340&-0.792240&-0.546200&-0.572940\\
                    -0.557160&-0.352250&0.773830&0.924710&-0.995350
                \end{bmatrix},$}
        \end{center}
    \end{samepage}
    with $x_5=[1.372851,\,-6.867020,\,-1.982185]^\top$. Disturbances are uniform in the ball of radius $\sqrt{0.0005}$,
    and random inputs are uniform in the admissible input ball. The initial margin is $\rho_5^\star=-1.44\times10^{-4}$. The number of subregions is $M_T=208$. %, and the SDP uses 208 covering regions.
    
    We fix $\rho_{\mathrm{tar}}=0.0300$ after separate preliminary trials and perform 100 new trials with a budget of ten additional
    inputs. For each method, we record
    \begin{align*}
        \tau(r)=\inf\{k\geq0:\rho_{5+k}^\star\geq r\}.
    \end{align*}
    Table~\ref{tab:samples} reports the results for this target, together with the lower level $r=0.02975$ evaluated from the same trajectories. At the lower level, both online policies succeed in all 100 trials, but the SDP has a lower  median number of samples,  4 against 7. At $r=0.0300$, the SDP reaches the target in 99 trials within five samples and in all 100 trials within 10
    samples. The SOCP does not attain such level within ten samples. Compared with random input, the SDP reduces the median from
    9 to 5 samples, a reduction of 44.4\%. 
    
    \begin{table}[!htbp]
        \centering
        \caption{Additional samples required to attain a prescribed stabilization margin.
            Medians use all 100 trials.}
        \label{tab:samples}
        \small
        \setlength{\tabcolsep}{4.3pt}
        \begin{tabular}{llccc}
            \toprule
            & & & \multicolumn{2}{c}{Trials attaining the target}\\
            \cmidrule(lr){4-5}
            Target $r$ & Method & Median & By $k=5$ & By $k=10$\\
            \midrule
            0.02975 & Random & 8 & 7/100 & 81/100\\
            & SOCP & 7 & 24/100 & 100/100\\
            & \textbf{SDP} & \textbf{4} & \textbf{100/100} & \textbf{100/100}\\
            \midrule
            0.0300 & Random & 9 & 5/100 & 76/100\\
            & SOCP & 130 & 0/100 & 0/100\\
            & \textbf{SDP} & \textbf{5} & \textbf{99/100} & \textbf{100/100}\\
            \bottomrule
        \end{tabular}
    \end{table}

    \section{Conclusions}
    We set up  a unifying framework for online input design for data-driven control of systems affected by noise, based on a formulation of a typical SDP arising in the informativity framework with QMI-representable noise constraints. We  defined a "quality measure" of the data associated with the variables characterizing such  SDP; we related its value on the data over $[0,T]$ to that over $[0,T+1]$; and we provided an iterative procedure to design the input at time $T$ so that the quality of the corresponding data at $T+1$ is not decreased. We showed the usefulness of our approach in dealing with six typical noisy data-driven control problems in the informativity framework, presenting numerical examples in which the procedure performs much better than one based on random inputs. 
    
    % Current research work aims at extending our framework to nonlinear systems, using reproducing kernel Hilbert spaces representations. 
    
    %\appendix
    \appendices
    
    \section{Three additional examples}\label{app:3moreex}

    \begin{example}[Additive feedback fragility]\label{ex:AFF}
        Assume that the data $\mathcal D_T$ are informative for quadratic stabilization and $\mathbf Q_T\succ0$. For $P\succ0$ and $\alpha\geq0$, denote by $\mathfrak K_T(P,\alpha)$  the set of stabilizing feedback gains for $\mathcal D_T$ with the same pair $(P,\alpha)$.
        Define
        \begin{align}
            \lambda_{\mathcal D_T}^{\rm a} := \sup\left\{ r\geq0: \exists\,P\succ0,\ \alpha\geq0,\ K \text{ such that } \right. \notag \\ 
            \left. K+\Delta_K\in\mathfrak K_T(P,\alpha), \forall\,\Delta_K\in\mathbb{R}^{m\times n}\text{ with }\|\Delta_K\|\leq r \right\}.
        \end{align}
        The quantity $\lambda_{\mathcal D_T}^{\rm a}$ is the supremal
        additive gain-perturbation radius that can be certified, through
        a common quadratic data-driven certificate, uniformly for all
        systems in $\Sigma_{T}$.
        
        To preserve affine dependence on the performance variable,
        we optimize the squared radius. Take $\vartheta=(Q,L)$ and define

        \begin{align}
            \mathcal F(Q,L,\rho) &:= \begin{bmatrix}
                Q & 0 & 0 & 0 & 0 \\
                0 & -Q & -L^{\top} & -Q & 0 \\
                0 & -L & -\rho I_m & 0 & L \\
                0 & -Q & 0 & I_n & Q \\
                0 & 0 & L^{\top} & Q & Q
            \end{bmatrix}, \\
            O &:= \begin{bmatrix}
                I_{2n+m} \\
                0_{2n\times(2n+m)}
            \end{bmatrix}.
        \end{align}
        Let
        \begin{align}
            \mathcal K_{\rm add} := \left\{ (Q,L,\alpha,\rho): Q=Q^{\top}\succeq0,\ \alpha\geq0,\ \rho\geq0 \right\}.
        \end{align}
        The corresponding optimal performance index is
        \begin{align}
            \rho_T^\star=\sup\left\{ \rho:\exists\ (Q,L,\alpha,\rho)\in\mathcal K_{\rm add}  \text{ such that } \right. \notag\\
            \left. \mathcal F(Q,L,\rho)-\alpha ON_TO^{\top}\succeq0,\right\}.
        \end{align}
        %According to literature, $\rho_T^\star=\bigl(\lambda_{\mathcal D_T}^{\rm a}\bigr)^2$. The result ensures that the supremal attainment is  available. Thus, the value $\sqrt{\rho_T^\star}$ is the largest certified radius for which one can choose a gain whose common quadratic stability guarantee is preserved for all $\|\Delta_K\|<\sqrt{\rho_T^\star}$ and all systems in $\Sigma_T$. 
        Under the stated assumptions, this SDP attains its optimal value
        and $\rho_T^\star=(\lambda_{\mathcal D_T}^{\rm a})^2$ (see Theorem~14 in \cite{li2025fragilityanalysisdatadrivenfeedback}). An optimizer determines the gain
        \begin{align}
            K^\star=L^\star(Q^\star)^\dagger,
        \end{align}
        for which $A+B(K^\star+\Delta_K)$ is Schur for every $(A,B)\in\Sigma_T$ and every $\|\Delta_K\|_2<\sqrt{\rho_T^\star}$. 
        \hfill $\blacksquare $
    \end{example} 
    \medskip
    \begin{example}[$\mathcal H_\infty$ control]\label{ex:Hinf}\rm 
        With the same performance output and $C_{Y,L}$, set $\rho=\gamma_\infty^{-2}$, take $\vartheta=(Y,L)$, and define
        \begingroup
        \small
        \setlength{\arraycolsep}{4pt}
        \begin{align}
            \mathcal F(Y,L,\rho) &:= \begin{bmatrix}
                Y & 0 & 0 & 0 & C_{Y,L}^{\top} \\
                0 & 0 & 0 & Y & 0 \\
                0 & 0 & 0 & L & 0 \\
                0 & Y & L^{\top} & Y-\rho I_n & 0 \\
                C_{Y,L} & 0 & 0 & 0 & I_\ell
            \end{bmatrix}, \\
            O &:= \begin{bmatrix}
                I_{2n+m} \\
                0_{(n+\ell)\times(2n+m)}
            \end{bmatrix}.
        \end{align}
        \endgroup
        Let
        \begin{align}
            \mathcal K_{\mathcal H_\infty} := \left\{ (Y,L,\alpha,\rho): Y=Y^{\top},\ \alpha\geq0,\ \rho\geq0 \right\}.
        \end{align}
        The corresponding optimal certifiable task level is
        \begin{align}
            \rho_T^\star = \sup\left\{ \rho: \exists\,(Y,L,\alpha,\rho)\in\mathcal K_{\mathcal H_\infty},\right.\notag\\ \left.\mathcal F(Y,L,\rho)-\alpha ON_TO^{\top}\succ 0\right\}.
        \end{align}
        A feasible solution with $\rho>0$ gives the gain $K=LY^{-1}$, which stabilizes all systems consistent with the data and
        guarantees a common $\mathcal H_\infty$ bound
        $\gamma_\infty=1/\sqrt{\rho}$. Hence, when $\rho_T^\star>0$, the infimum of these certified bounds is $1/\sqrt{\rho_T^\star}$. \hfill $\blacksquare $
    \end{example}
    \medskip 
    
    \begin{example}[Multiplicative feedback fragility]\label{ex:MFF}  
        The perturbed closed-loop matrix is $A+B\Delta_KK$, where $\Delta_K\in\mathbb R^{m\times m}$ satisfies $\|\Delta_K-I_m\|_2\le r$. Set $\rho=r^2$, take $\vartheta=(Q,L)$, and define 
        \begingroup
        \small
        \setlength{\arraycolsep}{4pt}
        \begin{align}
            \mathcal F(Q,L,\rho) &:= \begin{bmatrix}
                Q & 0 & 0 & 0 & 0 \\
                0 & -Q & 0 & -L^{\top} & 0 \\
                0 & 0 & (1-\rho)I_m & -I_m & 0 \\
                0 & -L & -I_m & I_m & L \\
                0 & 0 & 0 & L^{\top} & Q
            \end{bmatrix}, \\
            O &:= \begin{bmatrix}
                I_{2n+m} \\
                0_{(n+m)\times(2n+m)}
            \end{bmatrix}.
        \end{align}
        \endgroup
        Let
        \begin{align}
            \mathcal K_{\rm mult} := \left\{ (Q,L,\alpha,\rho): Q=Q^{\top}\succ0,\ \alpha\geq0,\ \rho\geq0 \right\}.
        \end{align}
        The corresponding performance index is
        \begin{align}
            \rho_T^\star = \sup\left\{ \rho: \exists\,(Q,L,\alpha,\rho)\in\mathcal K_{\rm mult},\right.\notag\\
            \left.\mathcal F(Q,L,\rho)-\alpha ON_TO^{\top}\succeq0\right\}.
        \end{align}
        For each feasible solution, set $K=LQ^{-1}$. The strict matrix inequality ensures that 
        \begin{align*}
            Q-(A+B\Delta_KK)Q(A+B\Delta_KK)^{\top}\succ0
        \end{align*}
        for all $(A,B)\in\Sigma_T$ and all $\|\Delta_K-I_m\|_2\le\sqrt{\rho}$.
        Thus, the same quadratic Lyapunov function $V(x)=x^{\top}Q^{-1}x$ certifies stability throughout this perturbation ball, and $\sqrt{\rho_T^\star}$ is the supremum of the radii certified by this condition. \hfill $\blacksquare$
    \end{example}
    
    \section{Three additional applications}\label{app:3morappl}
    
    \subsubsection{Additive feedback fragility}\label{sec:algoAFF} 
    %\begin{example}[{Additive feedback fragility}]\label{sec:algoAFF} 
    Suppose that the assumptions of Example~\ref{ex:AFF} hold for the current and updated data. Write $\lambda_T^{\rm a}:=\lambda_{\mathcal D_T}^{\rm a}$, so that $\rho_T^\star=(\lambda_T^{\rm a})^2$. Combining \eqref{eq65} and \eqref{eq:barrhoT} gives
    \begin{align*}
        \bigl(\lambda_{T+1}^{\rm a}\bigr)^2 \ge \bar\rho_T+\delta_T 
        \ge \bigl(\lambda_T^{\rm a}\bigr)^2 +\delta_T-\varepsilon_T.
    \end{align*}
    Thus, $\delta_T-\varepsilon_T$ is a lower bound on the increase in the squared certified radius. In particular, $\delta_T>\varepsilon_T$ guarantees a larger supremal certified additive perturbation radius.
    %\hfill $\blacksquare$
    %\end{example}
    %\medskip 
    
    \subsubsection{$\mathcal H_\infty$ control} 
    %\begin{example}[$\mathcal H_\infty$ control]
    The choice $\rho=\gamma_\infty^{-2}$ in Example~\ref{ex:Hinf} relates an increase in $\rho$ to a decrease in the certified
    $\mathcal H_\infty$ bound. For $\rho_T^\star>0$, set
    \begin{align*}
        \gamma_{\infty,T}:=\frac{1}{\sqrt{\rho_T^\star}},
    \end{align*}
    which is the infimum of the common bounds certified from the data.
    When $\delta_T\ge\varepsilon_T$, \eqref{eq65} and
    \eqref{eq:barrhoT} imply
    \begin{align*}
        \gamma_{\infty,T+1} \le \frac{1}{\sqrt{\bar\rho_T+\delta_T}}
        \le \frac{1}{\sqrt{\rho_T^\star+\delta_T-\varepsilon_T}} \le \gamma_{\infty,T}.
    \end{align*}
    The last comparison is strict if $\delta_T>\varepsilon_T$. Moreover, the corresponding matrices $Y$ and $L$ give the
    feedback gain $K=LY^{-1}$, which certifies the common bound $1/\sqrt{\bar\rho_T+\delta_T}$ for all systems consistent
    with the updated data.
    %\hfill $\blacksquare$
    %\end{example}
    %\medskip 
    
    \subsubsection{Multiplicative feedback fragility}\label{sec:algoMFF}
    %\begin{example}[Multiplicative feedback fragility]
    For Example~\ref{ex:MFF}, let $Q$ and $L$ be the matrices in the updated certificate at level $\bar\rho_T+\delta_T$. Recover the feedback gain and define
    \begin{align*}
        K=LQ^{-1}, \qquad V(x)=x^{\top}Q^{-1}x.
    \end{align*}
    The strict certificate ensures that
    \begin{align*}
        Q-(A+B\Delta_KK)Q(A+B\Delta_KK)^{\top}\succ0
    \end{align*}
    for every system consistent with the updated data and every multiplicative perturbation satisfying
    \begin{align*}
        \|\Delta_K-I_m\|_2 \le \sqrt{\bar\rho_T+\delta_T}.
    \end{align*}
    Consequently, $V$ is a common quadratic Lyapunov function for the perturbed closed-loop systems throughout this ball.
    %\hfill $\blacksquare$
    %\end{example}
    
    \section{Construction of a finite cover}\label{app:cover}
    Following the grid construction in the proof of Lemma~5 in \cite{Dumer2004}, we construct the regions in \eqref{eq:cover}
    by enclosing finitely many boxes in Euclidean balls. Fix $r_T>0$ and define, for $i=1,\ldots,n$,
    \begin{align*}
        b_i&:=\bar\sigma_T\sqrt{(\mathbf Q_T)_{ii}} +\sqrt{(\Omega_T)_{ii}},\\
        k_i&:=\max\left\{1,\left\lceil\frac{\sqrt n\,b_i}{r_T}\right\rceil\right\},\qquad M_T:=\prod_{i=1}^n k_i.
    \end{align*}
    The Cauchy--Schwarz inequality gives $\mathcal R_T\subseteq\prod_{i=1}^n[-b_i,b_i]$. Divide each interval $[-b_i,b_i]$ into $k_i$ equal subintervals and denote the resulting boxes by $\mathcal B_j=[\ell_j,u_j]$, $j=1,\ldots,M_T$. For each box, set
    \begin{align*}
        c_j:=\frac{\ell_j+u_j}{2},\qquad 
        a_j:=\frac{u_j-\ell_j}{2},\qquad
        D_j:=r_T I_n.
    \end{align*}
    For every $d\in\mathcal B_j$,
    \begin{align*}
        \|d-c_j\|_2^2 \le\|a_j\|_2^2
        =\sum_{i=1}^n\left(\frac{b_i}{k_i}\right)^2
        \le r_T^2.
    \end{align*}
    Hence $D_j\succ0$ and
    \begin{align*}
        \mathcal R_T\subseteq\bigcup_{j=1}^{M_T}\mathcal B_j \subseteq\bigcup_{j=1}^{M_T}\mathcal R_{T,j},
    \end{align*}
    which proves \eqref{eq:cover}. If $b_i=0$, take the corresponding subinterval to be $\{0\}$; the same argument applies.
    The same conversion applies to any finite box cover with side lengths at most $2r_T/\sqrt n$, with $M_T$ equal to the number of boxes.
    \bigskip

    \bibliographystyle{IEEEtran}
    \bibliography{References}
    
    \begin{IEEEbiography}[{\includegraphics[width=1in,height=1.25in,clip,keepaspectratio]{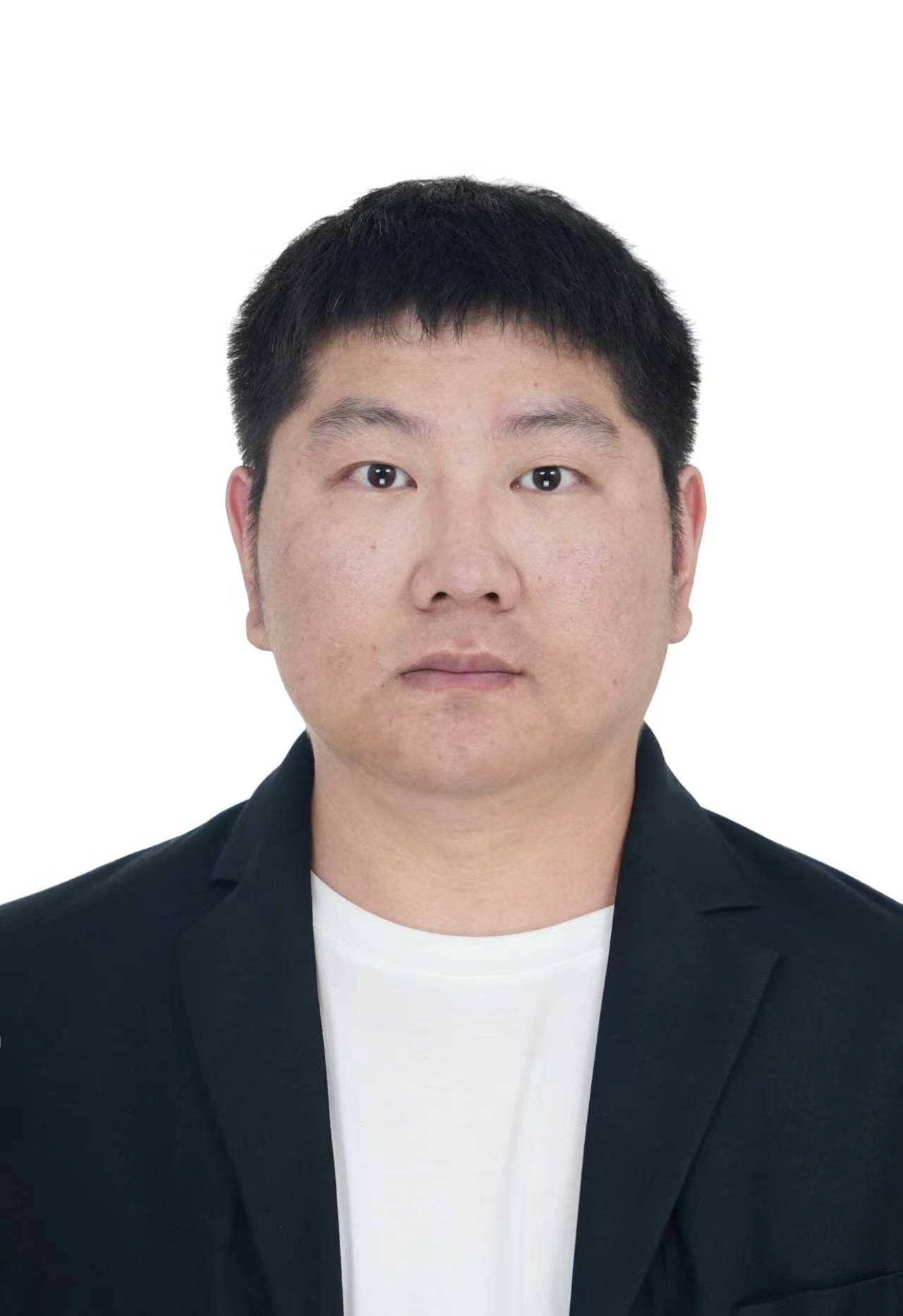}}]
        {\textbf{Yishu Wang}} (Member, IEEE) received the B.S. degree in mathematics from Xiangtan University, Xiangtan, China, in 2014, the M.S. degree and the Ph.D. degree in mathematics from Academy of Mathematics and Systems Science, Chinese Academy of Sciences, Beijing, China, in 2016 and 2019, respectively. He is now a postdoctoral researcher at Southeast University. His research interests include data-driven control, resilient control, nonlinear control, and multi-agent systems.
    \end{IEEEbiography}
    
    \begin{IEEEbiography}[{\includegraphics[width=1in,height=1.25in,clip,keepaspectratio]{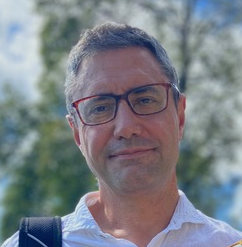}}]
        {\textbf{P. Rapisarda}}received his Ph.D. degree at the Department of Mathematics, University of Groningen, The Netherlands. He has held lecturing position at the University of Trieste, Italy; and at the University of Maastricht, The Netherlands. He is currently Professor in Control Theory at the School of Electronics and Computer Science of the University of Southampton, United Kingdom. His research interests include data-driven simulation and control, system identification, multidimensional systems, model reduction. He was {associate editor} of the {IEEE Transactions on Automatic Control} (2017-2022), of  {Systems and Control Letters}  (2004-16), and of {Multidimensional Systems and Signal Processing} (2012-2024). He is associate editor of the {IMA Journal of Mathematical Control and Information} since 2016 and of the IEEE CSS-Letters since 2022. 
    \end{IEEEbiography}
    
    \begin{IEEEbiography}[{\includegraphics[width=1in,height=1.25in,clip,keepaspectratio]{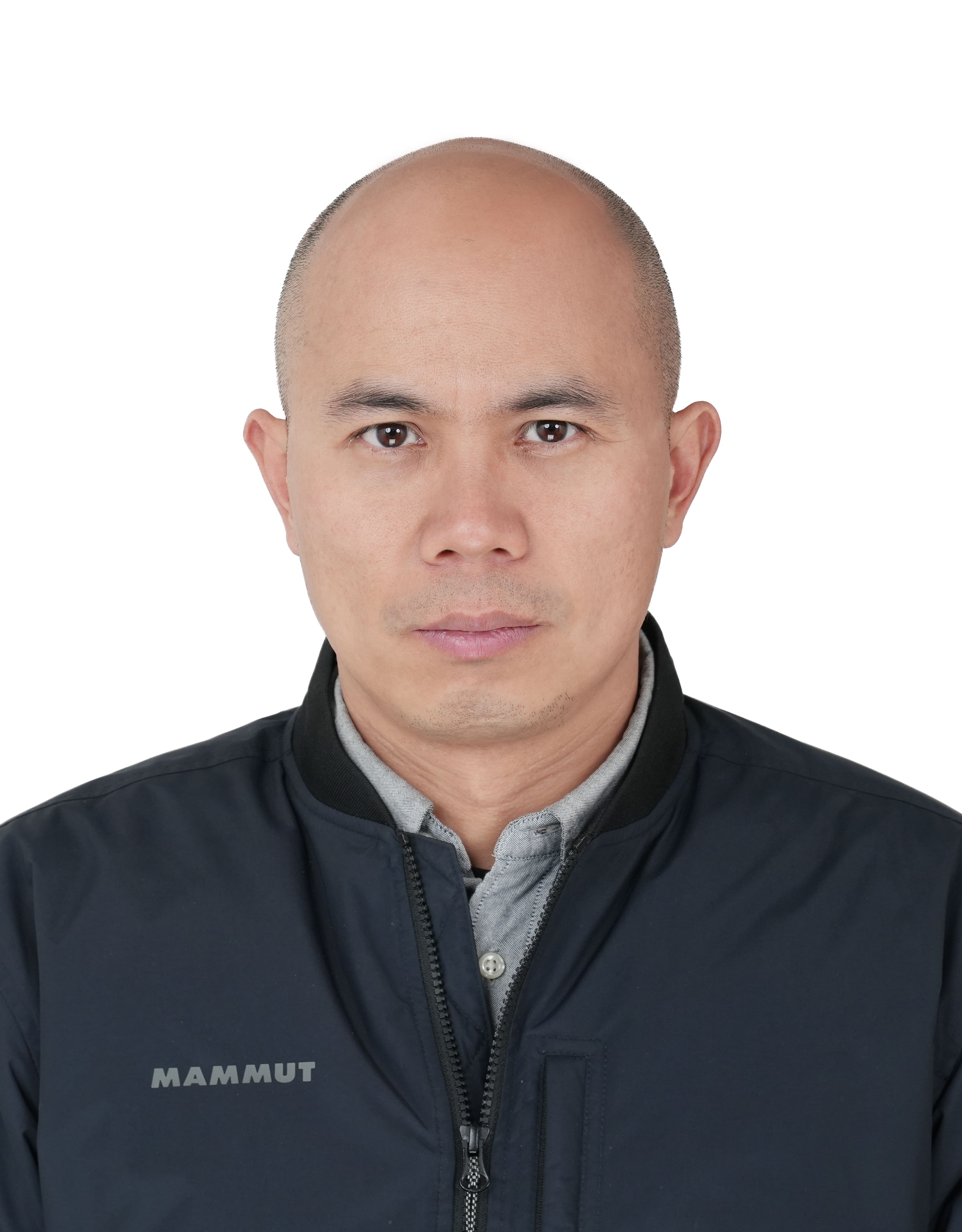}}]    	
        {\textbf{Jianquan Lu}} (Senior Member, IEEE) received the B.S. degree in mathematics from Zhejiang Normal University, Zhejiang, China, in 2003, the M.S. degree in mathematics from Southeast University, Nanjing, China, in 2006, and the Ph.D. degree in applied mathematics from City University of Hong Kong, Hong Kong, in 2009. From 2010 to 2012, he was an Alexander von Humboldt Research Fellow in PIK, Germany. He is currently a full professor with Southeast University, Nanjing, China. His current research interests include collective behavior in complex dynamical networks and multi-agent systems, logical networks, and hybrid systems. He has published over 100 IEEE Transactions and SIAM journal papers.
        
        Dr. Lu was named a Highly Cited Researcher by Clarivate Analytics, and he was elected Most Cited Chinese Researchers by Elsevier. Dr. Lu is an associate editor of \textit{Neural Processing Letters}, \textit{Journal of Franklin Institute}, \textit{Neural Computing and Applications}, and a guest editor of \textit{Science China: Information Sciences}, \textit{IET Control Theory \& Applications}.
    \end{IEEEbiography}

    \begin{IEEEbiography}[{\includegraphics[width=1in,height=1.25in,clip,keepaspectratio]{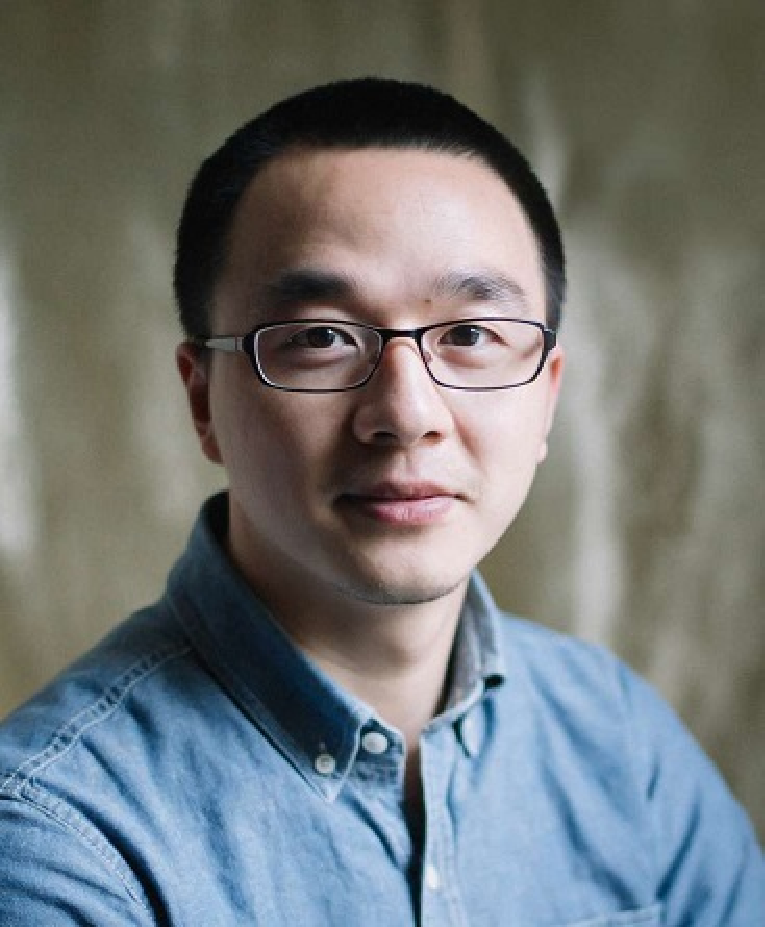}}]    	
        {\textbf{Yang Liu}} (Senior Member, IEEE)  received the B.S. degree in mathematics from Zhejiang Normal University, Zhejiang, China, in 2003, and the Ph.D. degree from Tongji University, Shanghai, in 2008. He is currently the dean of Hangzhou School of Automation and also a distinguished professor with School of Mathematical Sciences, Zhejiang Normal University. His research interests include logical systems, hybrid systems and distributed optimization. He has authored over 100 publications and three books. He is an IET Fellow, and he is an Associate Editor of Neural Processing Letters (Springer), Alexandria Engineering Journal, and Control and Decision. He was recognized by Elsevier as a Most Cited Chinese Researcher in 2020-2024, and by Clarivate Analytics as a Highly Cited Researcher in 2019-2022.
    \end{IEEEbiography}
    
\end{document}